\documentclass[11pt]{article}
\usepackage[margin=1in]{geometry}
\usepackage{amsmath,amssymb,amsthm,mathtools}
\usepackage{booktabs}
\usepackage{graphicx}
\usepackage{enumitem}
\usepackage{natbib}
\usepackage[colorlinks=true,linkcolor=blue,citecolor=blue,urlcolor=blue,hypertexnames=false]{hyperref}
\usepackage{microtype}
\usepackage{xurl}
\usepackage{float}
\usepackage{placeins}

\newtheorem{proposition}{Proposition}

\newtheorem{definition}{Definition}
\numberwithin{equation}{section}
\DeclareMathOperator*{\argmax}{arg\,max}
\newcommand{\dd}{\,d}
\newcommand{\E}{\mathbb E}

\graphicspath{{./}}

\title{Automation, Income Incidence, and Capital Accumulation in Incomplete Markets}
\author{Erhan Bayraktar\thanks{Department of Mathematics, University of Michigan. Email: \href{mailto:erhan@umich.edu}{erhan@umich.edu}. Supported in part by the National Science Foundation.}}
\date{July 2026}

\begin{document}
\maketitle

\begin{abstract}
This paper studies how automation changes the stationary distribution of income, consumption, and wealth in an incomplete-market economy. An automating sector trades off productivity gains and labor-cost savings against adoption costs. Households differ by skill and wealth, save in a capital/equity claim, and face uninsurable skill risk. Competitive factor prices and aggregate capital clear jointly with household Hamilton--Jacobi--Bellman equations and the stationary Kolmogorov forward equation. Automation affects consumption through labor-income incidence, precautionary saving, skill mobility, ownership of automation rents, and the stationary capital stock. In an adverse-incidence scenario with high exposure, adverse reskilling, capital obsolescence, and concentrated ownership, decentralized automation lowers stationary consumption and capital relative to the no-automation allocation. With stronger productivity and complementarity, lower obsolescence, and broader ownership, automation raises output, consumption, and capital. Reversing the assumed skill-mobility response also raises consumption, output, and capital substantially at a fixed automation intensity. A proxy diagnostic combining U.S. evidence on AI adoption, investment, labor-income pass-through, equity ownership, and marginal propensities to consume places the current economy near the boundary between the two scenarios. The model provides a quantitative framework for separating automation's aggregate gains from its distributional incidence.
\end{abstract}




\section{Introduction}

Automation affects more than production. It changes who receives labor income, who owns the resulting rents, and how much households save. In an incomplete-market economy those changes alter the stationary wealth distribution, the capital stock, and factor prices. This paper asks when productivity, task exposure, skill mobility, capital obsolescence, and rent ownership combine to raise or lower stationary consumption and capital.

The mechanism is an income-incidence channel in a real, flexible-price general equilibrium. Automation changes the allocation of output among consumption, replacement investment, adoption costs, and claims accruing outside the household sector. Labor-income risk and rent ownership then shape household consumption and saving:
\[
\begin{aligned}
\text{automation}
&\Longrightarrow \text{labor income, skill risk, and rent ownership}\\
&\Longrightarrow \text{household consumption and saving policies}\\
&\Longrightarrow \text{stationary wealth, capital, and factor prices}.
\end{aligned}
\]

The model combines an automating sector with a continuous-time incomplete-market household block. Households differ by wealth and by a two-state skill process. They choose consumption and saving, and their invariant distribution solves a KFE. A competitive final-good firm determines wages and the productive-capital return. Automation also generates rents; the share passed through to domestic asset holders determines the gap between the productive-capital return and the return received by households. The stationary equilibrium is computed by resolving the HJB--KFE and market-clearing system for each automation intensity.

The first contribution is to connect the task-based automation literature to heterogeneous-agent capital accumulation. Automation displaces some paid tasks and complements others, as in \citet{acemoglu2018,acemoglu2019,acemoglu2020wrong,acemoglu2020robots}. In Bewley--Huggett--Aiyagari economies, income risk and incomplete insurance determine saving and capital \citep{bewley1986,huggett1993,aiyagari1994}. Here the automation technology changes both the level and composition of labor income, so the stationary distribution becomes part of the incidence calculation. The ownership channel adds a second margin: households may lose labor income yet gain through claims on productive capital and automation rents.

In incomplete-market economies, pecuniary price effects can generate constrained inefficiency when they interact with missing markets or financial constraints. Establishing such inefficiency requires a welfare and feasibility benchmark of the kind developed by \citet{geanakoplospolemarchakis1986} and \citet{davilakorinek2018}. The present analysis instead characterizes the incidence of automation across stationary competitive equilibria.

The second contribution is a transparent decomposition of two scenario families. The \emph{adverse-incidence} scenario combines high low-skill exposure, automation-induced obsolescence, adverse reskilling, and concentrated ownership. The \emph{productivity-led capital-growth} scenario combines stronger productivity and high-skill complementarity with lower obsolescence and exposure. These are comparative-static scenarios, not two structural estimates. Their role is to show which assumptions account for the sign of the stationary consumption and capital responses.

The third contribution is computational. The numerical method follows the continuous-time HJB--KFE finite-difference approach of \citet{achdou2022}. The appendix states the finite-grid equilibrium actually solved, gives sufficient conditions for existence and conditional uniqueness on that grid, and reports market-clearing and residual checks. The contribution lies in the automation application, equilibrium decomposition, and reproducible scenario analysis. The empirical extension then maps the model's incidence terms into a transparent proxy diagnostic using current U.S. data.

The quantitative exercises use empirical guideposts for the skill wage ratio, AI exposure, and equity ownership. Automation elasticities governing reskilling and obsolescence remain scenario parameters, so the analysis reports neutral and reversed skill-mobility cases alongside the adverse-reskilling baseline. The reported allocations are stationary comparisons. A reduced-form policy index over consumption and exposed labor income organizes the instrument experiments; it is a political-economy criterion rather than a household-welfare ordering.

The paper proceeds as follows. Section 2 gives a diagnostic benchmark. Section 3 defines the stationary heterogeneous-agent equilibrium. Section 4 reports the numerical scenarios and robustness exercises. Section 5 studies tax, rebate, and ownership incidence. Section 6 develops the empirical regime diagnostic and policy experiment. Section 7 describes the replication code, and Section 8 concludes. The appendix documents the finite-difference method, finite-grid equilibrium result, and empirical validation exercises.

\section{Diagnostic Benchmark}

There is a unit mass of households split into skill groups \(s\in\{U,H\}\), where \(U\) denotes low-skilled or unskilled workers and \(H\) denotes high-skilled workers. The automation intensity is \(a\in[0,\bar a]\). Automation has three primitive effects. First, it raises Hicks-neutral productivity in the sense of a multiplicative technology shifter \citep{hicks1932,acemoglu2008},
\[
Z(a)=Z_0\exp(\psi_Za),\qquad \psi_Z>0.
\]
Second, it reduces paid low-skill human tasks and complements high-skill paid tasks,
\[
h_U(a)=\exp(-\chi_Ua),\qquad h_H(a)=\exp(\beta_Ha),
\]
with \(\chi_U>0\) and \(\beta_H>0\). Third, it changes production-task services,
\[
\ell_U(a)=\exp(-\xi_Ua),\qquad \ell_H(a)=\exp(\eta_Ha).
\]
The distinction between \(h_s\) and \(\ell_s\) is important. The firm pays wages on paid human tasks \(H\), but production uses total task services \(L\), which include automated task services. Define
\[
L(a)=\sum_s m_se_s\ell_s(a),\qquad H(a)=\sum_s m_se_sh_s(a).
\]
Output is
\[
Y(a,K)=Z(a)K^\alpha L(a)^{1-\alpha}.
\]
The efficiency-unit wage is
\[
w(a,K)=(1-\alpha)Z(a)K^\alpha L(a)^{-\alpha},
\]
and total collective wage bill is
\[
B(a,K)=w(a,K)H(a).
\]
The exposed low-skill collective wage bill is
\[
B_U(a,K)=w(a,K)m_Ue_Uh_U(a).
\]

A representative automating firm chooses \(a\) from
\[
\Pi(a;K,w)=Z(a)K^\alpha L(a)^{1-\alpha}-wH(a)-\phi a-\frac{\kappa}{2}a^2-P(a),
\]
where \(P(a)\) is a policy payment. Holding \((K,w)\) fixed, the private marginal benefit of automation is
\[
M(a;K,w)=
\underbrace{\left[\psi_Z+(1-\alpha)\frac{L_a(a)}{L(a)}\right]Y(a,K)}_{\text{productivity and production-task effect}}
+\underbrace{w[-H_a(a)]}_{\text{labor-cost effect}}.
\]
The private first-order condition is
\[
M(a;K,w)-\phi-\kappa a-P_a(a)=0.
\]
Let
\[
\iota=P_a(a)
\]
denote the marginal implementation wedge. A positive \(\iota\) discourages automation. A negative \(\iota\) subsidizes automation.

For comparison, define a reduced-form policy-maker index over stationary household-side outcomes:
\[
a^P\in\argmax_{a\in[0,1]}G_\mu(a),
\qquad
G_\mu(a)=\lambda C(a)+\mu B_U(a),
\]
where \(C(a)\) is household consumption and \(B_U(a)\) is labor income received by exposed low-skilled workers. The parameters \(\lambda,\mu\geq0\) are political-economy weights, not household preference parameters. Accordingly, \(a^P\) is only the maximizer of this index. It is not a planner allocation, and the ordering of \(a^P\) and the private choice has no efficiency content.

For any interior index target \(a^P\), the marginal payment that implements it in the firm's static first-order condition is
\[
\iota^*=M(a^P;K,w)-\phi-\kappa a^P.
\]
At a boundary target \(a^P=0\), any tax above the minimal deterrent value \(\bar\tau=M(0;K,w)-\phi\) implements the boundary.

The diagnostic comparison is summarized by the derivative
\[
G_\mu'(a)=\lambda C_a(a)+\mu B_{U,a}(a).
\]
When the private residual \(F(a)=M(a;K,w)-\phi-\kappa a\) has a unique root, the sign of this derivative at that root indicates whether the index continues to rise or has begun to fall. Appendix \ref{app:diagnostic_benchmark} decomposes the derivative into low- and high-skill exposure, mobility, productivity, and depreciation terms. The decomposition organizes the channels before prices, saving, and the wealth distribution are allowed to adjust.

\section{Stationary Heterogeneous-Agent General Equilibrium}

\subsection{Households}

Time is continuous. A unit mass of households has wealth \(k\geq \underline k\) and skill state \(s\in\mathcal S=\{U,H\}\). The skill state follows
\[
Q(a)=
\begin{pmatrix}
-q_{UH}(a)&q_{UH}(a)\\
q_{HU}(a)&-q_{HU}(a)
\end{pmatrix},
\qquad
q_{UH}(a)=q_0e^{-\zeta a},
\qquad
q_{HU}(a)=q_0e^{\zeta a}.
\]
The baseline \(\zeta>0\) is an adverse-reskilling scenario: automation lowers the transition rate from \(U\) to \(H\) and raises the reverse rate. This direction is not imposed as an empirical fact. Section \ref{subsec:mobility_robustness} reports \(\zeta=0\) and \(\zeta<0\), the latter representing AI-assisted upward mobility.

Let \(g_s(k)\) be the stationary density. Aggregate capital is
\[
K=\sum_s\int kg_s(k)\dd k.
\]
Lowercase \(k\) is an individual household's asset state; uppercase \(K\) is the integral of those assets across the stationary distribution and equals productive capital in equilibrium.
Production-task labor and paid human-task labor are
\[
L(a;g)=\sum_s\int e_s\ell_s(a)g_s(k)\dd k,
\qquad
H(a;g)=\sum_s\int e_sh_s(a)g_s(k)\dd k.
\]
The wage-cost exposure index is
\begin{equation}\label{eq:lambda_h}
\Lambda_H(a;g)=-H_a(a;g)=-\sum_s\int e_sh_s'(a)g_s(k)\dd k.
\end{equation}
A household in state \(s\) receives labor income
\[
y_s^L(a)=w e_sh_s(a).
\]
The variable \(w\) is a wage per efficiency unit, not a literal common wage paid to every worker. High-skilled households can earn more because \(e_H>e_U\) and because AI raises \(h_H(a)\).

\subsection{Prices and market clearing}

The final-good firm is competitive and has Cobb--Douglas technology
\[
Y=Z(a)K^\alpha L^{1-\alpha}.
\]
Given \((K,L,a)\), it solves
\[
\max_{K,L}\left\{Z(a)K^\alpha L^{1-\alpha}-(r+\delta(a))K-wL\right\},
\]
where
\[
\delta(a)=\delta_0+\delta_Aa.
\]
The term \(\delta_Aa\) captures obsolescence and reorganization of legacy physical capital when AI is adopted. This is why the model can generate lower stationary \(K\) in the high-automation equilibrium even though automation raises \(Z(a)\). The first-order conditions are
\[
r=\alpha Z(a)K^{\alpha-1}L^{1-\alpha}-\delta(a),
\qquad
w=(1-\alpha)Z(a)K^\alpha L^{-\alpha}.
\]
Constant returns imply zero pure profits for the final-good firm:
\[
Y-(r+\delta(a))K-wL=0.
\]
Thus \(r\) is the real return on productive capital, or the investment opportunity. The price of installed capital is normalized to one, so household wealth and aggregate capital are measured in units of the final good. Consequently, the model has a real capital return \(r\), but no separate Tobin's \(q\) or stock-market price for capital \citep{tobin1969}.

Automation creates a rent by reducing paid human tasks below total production-task services. Let
\[
\Phi^A(a)=\phi a+\frac{\kappa}{2}a^2
\]
be the real resource cost of automation. Net automation rents are
\[
\Pi^A(a;g)=w[L(a;g)-H(a;g)]-\Phi^A(a)-\tau a.
\]
A share \(\theta_E\in[0,1]\) is owned by domestic households, so the household asset return is
\[
R=r+\theta_E\frac{\Pi^A(a;g)}{K}.
\]
A low-skilled but wealthy household can therefore be partly insulated from automation through capital and automation-rent income.

If the policy is a literal tax, fiscal closure is required. The firm pays \(\tau a\). A fraction \(\omega_T\) is lost to administrative or political-economy frictions, and the remaining part is rebated through a schedule \(T_s(k)\):
\begin{equation}
T_s(k)=(1-\omega_T)\tau a\, b_s(k),
\qquad
b_s(k)\geq0,
\qquad
\sum_s\int b_s(k)g_s(k)\dd k=1.
\label{eq:rebate_schedule}
\end{equation}
The function \(b_s(k)\) is the rebate kernel: \(b_s(k)=1\) gives a lump-sum rebate, while kernels increasing or decreasing in income generate proportional or progressive rebates. Here \(T_s(k)\) is the household rebate generated by the tax revenue.

Given a stationary environment \((a,w,R,T,Q)\), a household starts from \((k,s)\) and solves
\[
V_s(k)=
\sup_{\{c_t\}_{t\geq0}}
\E_{k,s}\left[\int_0^\infty e^{-\rho t}u(c_t)\dd t\right]
\]
subject to
\[
\dot k_t=Rk_t+w e_{s_t}h_{s_t}(a)+T_{s_t}(k_t)-c_t,
\qquad k_t\geq\underline k,
\]
and the Markov transition matrix \(Q(a)\). The corresponding stationary HJB is
\[
\rho V_s(k)=\max_{c\geq0}\left\{u(c)+V_s'(k)[Rk+w e_sh_s(a)+T_s(k)-c]+\sum_{s'}q_{ss'}(a)[V_{s'}(k)-V_s(k)]\right\},
\]
where
\[
u(c)=\frac{c^{1-\gamma}}{1-\gamma}.
\]
The optimality condition and savings drift are
\[
u'(c_s(k))=V_s'(k),
\qquad
\dot k_s(k)=Rk+w e_sh_s(a)+T_s(k)-c_s(k).
\]
Thus the law of motion for \(k\) is the household budget constraint evaluated at the optimal consumption policy. The productive-capital return \(r\) is the marginal product net of depreciation. The household return \(R\) additionally includes the domestically owned share of automation rents; hence \(R\) and \(r\) coincide only when that rent yield is zero.
The stationary Kolmogorov forward equation is
\[
0=-\partial_k[\dot k_s(k)g_s(k)]+
\sum_{s'}q_{s's}(a)g_{s'}(k)-\sum_{s'}q_{ss'}(a)g_s(k).
\]

The stationary goods-market-clearing condition is
\begin{equation}
\boxed{
\mathcal E^{goods}\equiv Y-\left[C+\delta(a)K+\Phi^A(a)+\omega_T\tau a+(1-\theta_E)\Pi^A(a;g)\right]=0.
}
\label{eq:goods_residual}
\end{equation}
The residual in \eqref{eq:goods_residual} is the Walrasian accounting identity for the one-good economy \citep{arrowdebreu1954,mascolell1995}. Once household budgets, firm zero profit, the government budget, and capital-market clearing hold, the goods market clears by Walras' law; because the model has automation costs, fiscal leakage, and foreign ownership, the residual is reported explicitly in the numerical results.

\subsection{Firm automation and the stationary fixed point}

Taking prices and the household distribution as given, the automating sector chooses \(a\) from
\[
\Pi(a)=Z(a)K^\alpha L(a;g)^{1-\alpha}-wH(a;g)-\Phi^A(a)-P(a).
\]
The private marginal benefit of automation, before policy, is
\[
\mathcal M(a;g,K,w)
=
\left[\psi_Z+(1-\alpha)\frac{L_a(a;g)}{L(a;g)}\right]Y+w\Lambda_H(a;g).
\]
Thus the firm's automation first-order condition is
\begin{equation}
\mathcal M(a;g,K,w)-\phi-\kappa a-\iota=0,
\qquad \iota=P_a(a).
\label{eq:firm_automation_foc}
\end{equation}
The derivative in \eqref{eq:firm_automation_foc} is the sector's direct profit derivative at fixed \((g,K,w)\), not the total derivative of aggregate consumption or prices. The decentralized economy has \(\iota=0\). Numerically, for each candidate \(a\) the household and market-clearing block is resolved, the direct profit derivative is evaluated at those stationary objects, and \(a^{D}\) is a root of \eqref{eq:firm_automation_foc}.

\begin{definition}[Stationary competitive equilibrium]\label{def:stationary_equilibrium}
For a fixed policy rule, a stationary equilibrium is a collection
\[
\{a,r,w,R,K,V_s,c_s,g_s,T_s\}_{s\in\{U,H\}}
\]
such that: (i) \(V_s\) and \(c_s\) solve each household's infinite-horizon problem; (ii) \(g_s\) is invariant under the induced savings drift and skill transitions; (iii) the final-good firm satisfies its static factor-demand conditions; (iv) household asset supply equals productive capital \(K\); (v) fiscal and goods-market accounts clear; and (vi) \(a\) satisfies the automating sector's private optimality condition, including any policy payment.
\end{definition}

Households continue to move across wealth and skill states in a stationary equilibrium, while the cross-sectional distribution remains invariant. Appendix \ref{app:finite_grid_equilibrium} establishes existence and conditional uniqueness for the finite-grid system used in the computation. Transition dynamics and dynamic policy are natural extensions of the stationary analysis.

\subsection{Policy-maker index and equilibrium derivative}
\label{subsec:stationary_marginal_derivative}

This subsection defines the reduced-form index used in the quantitative comparisons. The full stationary economy determines prices, capital, consumption policies, skill masses, and the wealth distribution jointly. For each automation value \(a\), write
\[
g^a,K(a),L(a),H(a),C(a),w(a),r(a),R(a)
\]
for the corresponding market-clearing objects. The policy-maker index is
\[
G_\mu^{\mathrm{stat}}(a)=\lambda C(a)+\mu B_U(a),
\qquad
B_U(a)=w(a)\int e_Uh_U(a)g_U^a(k)\,dk,
\]
where \(B_U(a)\) is the collective wage bill paid to exposed low-skilled households. The marginal object is therefore
\[
\frac{dG_\mu^{\mathrm{stat}}(a)}{da}=\lambda C_a(a)+\mu B_{U,a}(a).
\]
This derivative is the bridge between the diagnostic benchmark and the computation. It ranks stationary allocations under the chosen political-economy weights. If the index target \(a^P\) is interior, the wedge that implements it is obtained by evaluating the private condition \eqref{eq:firm_automation_foc} at the target:
\[
\iota^*=\mathcal M(a^P;g^{a^P},K(a^P),w(a^P))-\phi-\kappa a^P.
\]
If \(a^P=0\), the minimal boundary tax is
\[
\bar\tau=\mathcal M(0;g^0,K(0),w(0))-\phi.
\]

The derivative \(C_a(a)\) is a total equilibrium derivative, not a partial derivative holding prices or the distribution fixed:
\[
C_a(a)=
\sum_s\int \frac{\partial c_s(k,a)}{\partial a}g_s^a(k)\,dk
+
\sum_s\int c_s(k,a)\frac{\partial g_s^a(k)}{\partial a}\,dk.
\]
The first term is the change in consumption policies at a fixed distribution. The second term is the change in the stationary distribution. The numerical section evaluates this total derivative by resolving the stationary market-clearing HJB--KFE system across automation values.

Appendix \ref{app:diagnostic_benchmark} gives a stripped-down no-wealth-heterogeneity benchmark in which \(K(a)\), \(w(a)\), and the skill masses can be eliminated analytically. That benchmark provides closed-form sign conditions and predicted automation levels. The finite-grid heterogeneous-agent equilibrium used in the numerical section is treated separately in Appendix \ref{app:finite_grid_equilibrium}.

\section{Stationary Numerical Results}

This section organizes the quantitative evidence into four related exercises. The baseline describes the technology and market-clearing mechanics. The regime comparison contrasts adverse-incidence with productivity-led capital growth. The incidence results report wealth distributions, consumption functions, and stationary consumption-equivalent comparisons. Section 5 then studies fiscal instruments and ownership. Figures \ref{fig:z_scenarios}--\ref{fig:prices} show the technology, labor-income, and market-clearing mechanics. Figure \ref{fig:goods} and Tables \ref{tab:resource_grid}--\ref{tab:computed_regimes} show goods-market accounting and the two AI regimes. Figures \ref{fig:density}--\ref{fig:consumption_gain} show distributional incidence and consumption smoothing in the adverse-incidence baseline. The productivity-led figures in Subsection \ref{subsec:capital_growth_case} repeat the same analysis for the capital-growth case.

\subsection{Calibration and data discipline}

The baseline calibration is
\[
\alpha=0.36,\quad \delta_0=0.06,\quad \delta_A=0.25,\quad Z_0=1,
\quad \psi_Z=0.18,
\]
\[
\rho=0.15,\quad \gamma=2,
\quad \phi=0.01,
\quad \kappa=0.52.
\]
Skill efficiencies are
\[
e_U=0.75,
\qquad e_H=1.25,
\]
so \(e_H/e_U=1.67\), close to the BLS and Pew high-/low-wage ratios discussed in the introduction. Task functions are
\[
h_U(a)=e^{-3.20a},\qquad h_H(a)=e^{0.35a},
\]
\[
\ell_U(a)=e^{-2.50a},\qquad \ell_H(a)=e^{0.55a}.
\]
Thus automation both displaces low-skilled paid tasks and complements high-skilled labor. This sign combination creates a private benefit for high-skill users and a strong adoption incentive; it is one scenario rather than a necessary feature of automation. High-skilled workers nevertheless remain a minority in the computed high-automation economy. The automation-dependent transition rates make the stationary mass of high-skilled households about 0.313 at the decentralized allocation, while the low-skilled mass is about 0.687. Hence high-skilled workers benefit per capita, but there are fewer of them and they do not form the main consumption base. The cost parameters are selected so that the adverse-incidence scenario produces a decentralized automation index near one half. The index is a task-technology intensity. It should not be read as a literal employment share.

The exponential specifications are reduced-form constant-semi-elasticity functions. They preserve positivity and make the exposure parameters transparent; they are not derived from an organizational model of AI adoption. The signs \(h_U',\ell_U'<0\) and \(h_H',\ell_H'>0\) define a displacement/complementarity scenario, while their magnitudes are selected scenario inputs. The same is true of \(\delta_A\) and \(\zeta\). The baseline jointly selects strong low-skill exposure, adverse reskilling, and legacy-capital obsolescence, so it should be interpreted as an adverse-incidence stress test rather than a central estimate.

Table \ref{tab:all_parameters} lists the full baseline calibration. The symbol \(\theta_E\) denotes domestic ownership/pass-through of automation rents. The baseline policy has \(\tau=0\). Positive taxes are introduced only in the fiscal experiments in Section~5.

\begin{table}[!htbp]
\centering
\caption{Full baseline parameter list}
\label{tab:all_parameters}
\small
\begin{tabular}{p{0.16\textwidth}p{0.56\textwidth}p{0.20\textwidth}}
\toprule
Symbol & Meaning & Baseline value \\
\midrule
\(\alpha\) & capital share in final-good production & 0.36 \\
\(Z_0\) & baseline productivity & 1.00 \\
\(\psi_Z\) & productivity gain from automation & 0.18 \\
\(\delta_0\) & baseline depreciation & 0.06 \\
\(\delta_A\) & automation-induced obsolescence of legacy capital & 0.25 \\
\(\rho\) & household discount rate & 0.15 \\
\(\gamma\) & coefficient of relative risk aversion & 2.00 \\
\(\phi\) & linear resource cost of automation & 0.01 \\
\(\kappa\) & convex resource cost of automation & 0.52 \\
\(e_U,e_H\) & skill efficiency levels & (0.75, 1.25) \\
\(\chi_U,\beta_H\) & paid-task exposure/complementarity in \(h_s(a)\) & (3.20, 0.35) \\
\(\xi_U,\eta_H\) & production-task exposure/complementarity in \(\ell_s(a)\) & (2.50, 0.55) \\
\(q_0,\zeta\) & transition-rate level and automation tilt & (0.50, 0.75) \\
\(\lambda\) & policy-index weight on aggregate consumption & 0.60 \\
\(\mu\) & policy-index weight on exposed low-skill labor income & 1.00 \\
\(\theta_E\) & domestic ownership/pass-through share of automation rents & 0.45 \\
\(\omega_T\) & fiscal friction/loss share from tax collection & 0.15 \\
\(\tau\) & automation tax in baseline / fiscal experiments & 0 / varied \\
\bottomrule
\end{tabular}
\end{table}

The parameter \(\theta_E\) deserves special interpretation. It is not the fraction of households who own stocks. It is the share of net automation rents that is passed through to domestic households through capital and equity claims. The value \(\theta_E=0.45\) is the central pass-through benchmark. The value \(\theta_E=0.15\) is retained as a low-pass-through robustness case, while \(\theta_E=0.75\) is a high domestic-ownership case. This distinction matters because U.S. equity ownership is domestically large but highly concentrated; moreover, automation rents may be retained inside firms, capitalized into prices, accrue to managers and suppliers, or be owned by foreign investors. The calibration therefore separates the aggregate ownership share from the distribution of ownership across wealth groups.

\paragraph{Interpreting the productivity shifter.}
The choice of \(Z(a)\) should be read as a calibrated technology shifter, not as a direct estimate of the mechanical productivity effect of every AI application. With
\[
Z(a)=Z_0e^{\psi_Za},
\]
a full-index productivity gain of \(g_Z\) corresponds to
\[
\psi_Z=\log(1+g_Z).
\]
For example, a 15 percent full-adoption productivity gain corresponds to \(\psi_Z\simeq0.14\). The baseline value \(\psi_Z=0.18\) therefore implies a roughly 20 percent full-index gain. That is deliberately moderate relative to highly optimistic AI forecasts. Goldman Sachs Research has estimated that generative AI could raise global GDP by about 7 percent and, in later labor-market work, that full adoption could raise labor productivity in developed markets by around 15 percent. McKinsey estimates that generative AI could add \$2.6--\$4.4 trillion annually and contribute additional productivity growth depending on adoption and redeployment of labor \citep{goldman2023gdp,goldman2026labor,mckinsey2023genai}. Figure \ref{fig:z_scenarios} therefore reports \(Z(a)\) under three interpretations: an empirical guidepost, the baseline calibration, and a productivity-led counterfactual.

\begin{figure}[H]
\centering
\includegraphics[width=0.82\textwidth]{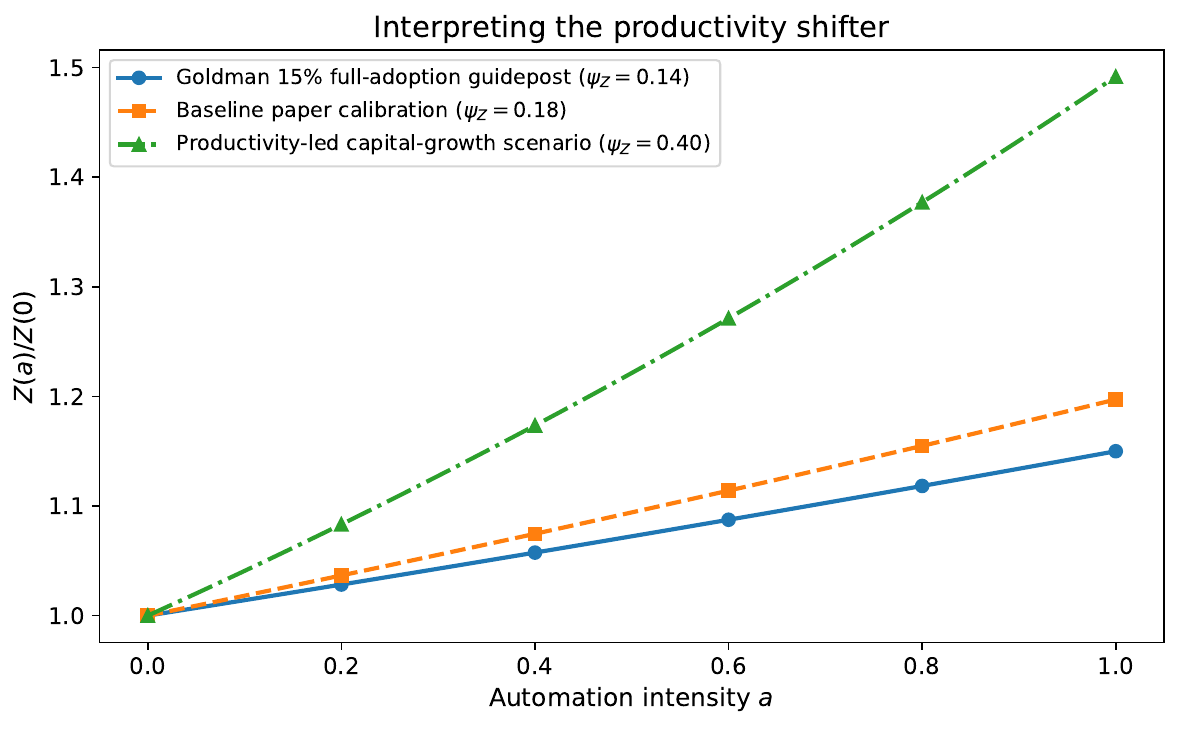}
\caption{Interpreting \(Z(a)=Z_0e^{\psi_Za}\). The baseline productivity gain is moderate; the productivity-led scenario deliberately allows stronger output gains.}
\label{fig:z_scenarios}
\end{figure}

The baseline political-economy weights are \(\lambda=0.6\) and \(\mu=1\), and the calculations also report \(\mu=0\). These weights affect only the reduced-form policy index; they do not enter household utility or the decentralized equilibrium.

Before Table \ref{tab:data_targets}, the marginal propensity to consume (MPC) terminology is defined mathematically. The reported MPC is the local wealth derivative of the consumption policy,
\begin{equation}
\mathrm{MPC}_s^k(k,a)=\frac{\partial c_s(k,a)}{\partial k}.
\label{eq:wealth_mpc}
\end{equation}
Thus a ``high-MPC exposed group'' means that households losing labor income have steep consumption functions near the borrowing constraint. Labor-income and policy changes are evaluated by resolving the full stationary equilibrium.

\begin{table}[!htbp]
\centering
\caption{Empirical targets and model discipline}
\label{tab:data_targets}
\small
\begin{tabular}{p{0.16\textwidth}p{0.50\textwidth}p{0.24\textwidth}}
\toprule
Object & Empirical discipline & Model use \\
\midrule
Skill wage ratio & BLS reports 2024 median weekly earnings of \$1,543 for bachelor's-degree workers and \$930 for high-school workers. Pew reports \$33/hour in high-AI-exposure jobs versus \$20/hour in least-exposed jobs. & Set \(e_H/e_U=1.67\). \\
AI complementarity & Pew reports that AI exposure is concentrated among higher-paid, college-educated workers. & Motivates, but does not identify, the scenario \(h_H'(a)>0\) and \(\ell_H'(a)>0\). \\
High-MPC exposure & Fiscal evidence emphasizes large MPCs among liquidity-constrained and hand-to-mouth households. & Low-wealth households have high local MPCs from the HJB solution. \\
Equity ownership & The Federal Reserve Distributional Financial Accounts track concentrated corporate-equity ownership. & Vary domestic ownership \(\theta_E\). \\
Skill mobility & No causal target is imposed for the effect of AI on reskilling. & Report adverse, neutral, and reversed mobility tilts. \\
Capital obsolescence & No direct estimate disciplines \(\delta_A\). & Vary \(\delta_A\) across stress-test and productivity-led scenarios. \\
\bottomrule
\end{tabular}
\end{table}

\subsection{Skill-mobility robustness}\label{subsec:mobility_robustness}

The sign of the skill-transition response is consequential. Table \ref{tab:mobility_robustness} holds automation fixed at the baseline decentralized intensity \(a=0.526\) and resolves the stationary equilibrium under three transition tilts. The baseline \(\zeta=0.75\) shifts mass toward the low-skill state. The neutral case sets \(\zeta=0\). The reversed case \(\zeta=-0.75\) represents AI-assisted upward mobility. All other parameters are unchanged.

\begin{table}[!htbp]
\centering
\caption{Skill-mobility robustness at fixed automation \(a=0.526\)}
\label{tab:mobility_robustness}
\begin{tabular}{lccccccc}
\toprule
Mobility tilt & \(\zeta\) & \(m_U\) & \(m_H\) & \(K\) & \(C\) & \(Y\) & \(r\) \\
\midrule
Reversed / upward & -0.75 & 0.312 & 0.688 & 2.896 & 1.162 & 1.822 & 0.0350 \\
Neutral & 0.00 & 0.500 & 0.500 & 2.687 & 0.889 & 1.503 & 0.0099 \\
Adverse baseline & 0.75 & 0.688 & 0.312 & 2.036 & 0.609 & 1.088 & 0.0009 \\
\bottomrule
\end{tabular}
\begin{minipage}{0.94\textwidth}
\vspace{0.35em}
\footnotesize
Each row resolves the household HJB--KFE and market-clearing system at the same \(a\). Absolute capital-market residuals are below \(5\times10^{-5}\), and absolute goods-market residuals are below \(3\times10^{-7}\).
\end{minipage}
\end{table}

Reversing the mobility tilt substantially raises the high-skill mass, consumption, output, and capital at the same automation intensity. The adverse transition matrix is therefore one driver of the baseline outcome, not an innocuous normalization. This sensitivity is why the quantitative results are presented as scenario comparisons rather than as a general prediction about AI.

\subsection{Baseline equilibrium and why firms automate}

For any policy \(\tau\), the automation number reported in the decentralized row is endogenous conditional on the scenario parameter vector. It is the firm's profit-maximizing stationary choice after the household distribution and interest rate have cleared. I denote it by \(a^{D}(\tau)\). In the no-tax baseline, \(a^{D}=a^{D}(0)\) solves the firm condition \eqref{eq:firm_automation_foc}:
\[
\mathcal E_a(a,0)
=
\left[\psi_Z+(1-\alpha)\frac{L_a(a;g^a)}{L(a;g^a)}\right]Y(a)
+w(a)\Lambda_H(a;g^a)-\phi-\kappa a=0.
\]
Here \(g^a\), \(K(a)\), \(w(a)\), and \(r(a)\) are the stationary market-clearing objects associated with that value of automation. The superscript \(D\) denotes the decentralized stationary Aiyagari/Huggett allocation, not a planner choice. Table \ref{tab:baseline_ge} compares it with the maximizer of the reduced-form policy index.

\begin{table}[!htbp]
\centering
\caption{Stationary allocations in the adverse-incidence scenario}
\label{tab:baseline_ge}
\resizebox{\textwidth}{!}{%
\begin{tabular}{lcccccccccc}
\toprule
Allocation & $a$ & $K$ & $L$ & $H$ & $r$ & $R$ & $w$ & $Y$ & $C$ & $B$ \\
\midrule
Decentralized & 0.526 & 2.036 & 0.674 & 0.565 & 0.001 & 0.006 & 1.055 & 1.088 & 0.609 & 0.597 \\
Policy-index target & 0.000 & 2.540 & 1.000 & 1.000 & 0.138 & 0.138 & 0.895 & 1.399 & 1.246 & 0.895 \\
\bottomrule
\end{tabular}}
\begin{minipage}{0.93\textwidth}
\vspace{0.35em}
\footnotesize
The policy index uses \((\lambda,\mu)=(0.6,1)\). The minimal boundary tax that implements its \(a=0\) target is \(\bar\tau=0.589\), and it collects zero revenue in that stationary allocation. The absolute goods-market residual is below \(10^{-7}\) in both rows.
\end{minipage}
\end{table}

The decentralized firm chooses \(a^{D}=0.526\), while the policy index is maximized at \(a^P=0\) under the reported weights. Figure \ref{fig:skill_income} shows that high-skilled labor income rises sharply with automation. At the decentralized allocation, low-skilled labor income per worker is about \(0.147\), while high-skilled labor income is about \(1.586\). At \(a=0\), the corresponding values are \(0.671\) and \(1.119\). The direct productivity, complementarity, and labor-cost terms explain the private choice; the stationary comparison records the accompanying incidence across skill and wealth groups.

In this adverse-incidence scenario, moving from \(a=0\) to the decentralized intensity lowers low-skilled income, stationary consumption, and the capital stock. Capital falls from \(2.540\) to \(2.036\). This comparison reflects several simultaneous assumptions: automation raises \(Z(a)\), but it also reduces production-task labor, increases effective depreciation through \(\delta(a)\), and shifts the skill distribution toward \(U\). It should not be attributed to a standalone aggregate-demand force.

\begin{figure}[H]
\centering
\includegraphics[width=0.82\textwidth]{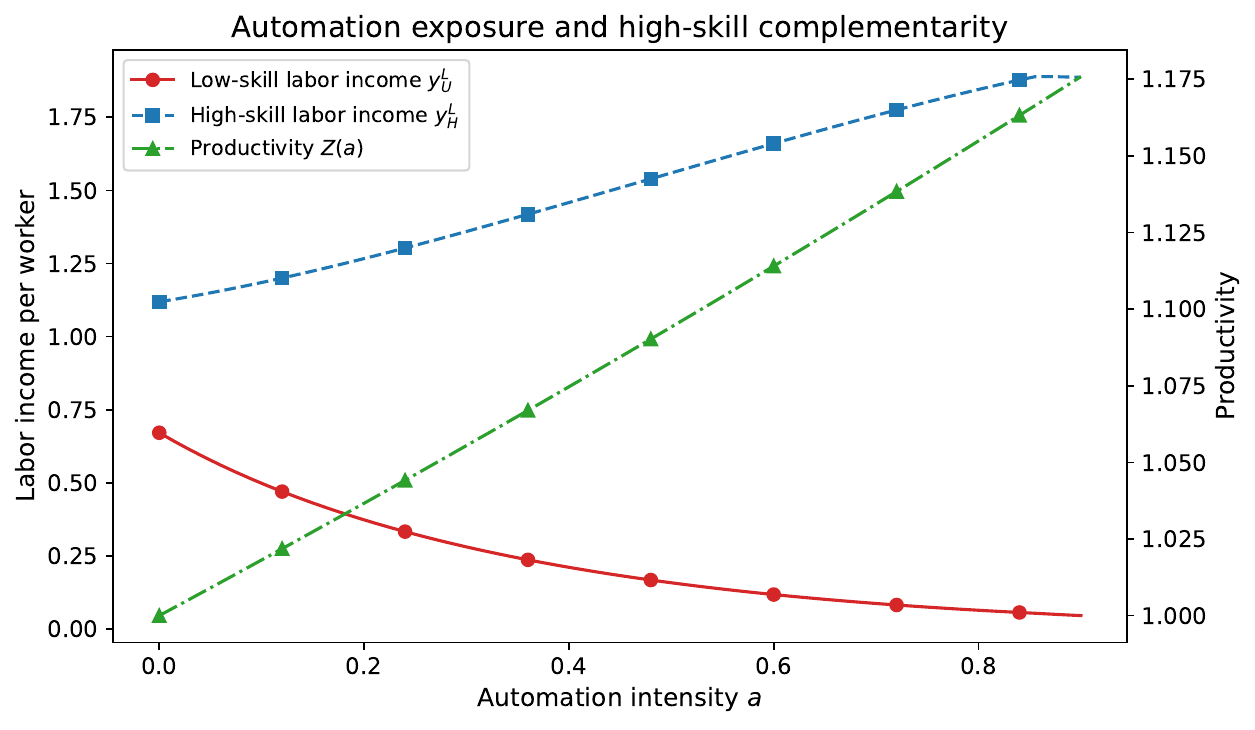}
\caption{Automation exposure and high-skill complementarity. Low-skilled labor income falls with automation, while high-skilled labor income rises. This is the private benefit that makes automation attractive to firms and high-skilled workers.}
\label{fig:skill_income}
\end{figure}

\begin{figure}[H]
\centering
\includegraphics[width=0.82\textwidth]{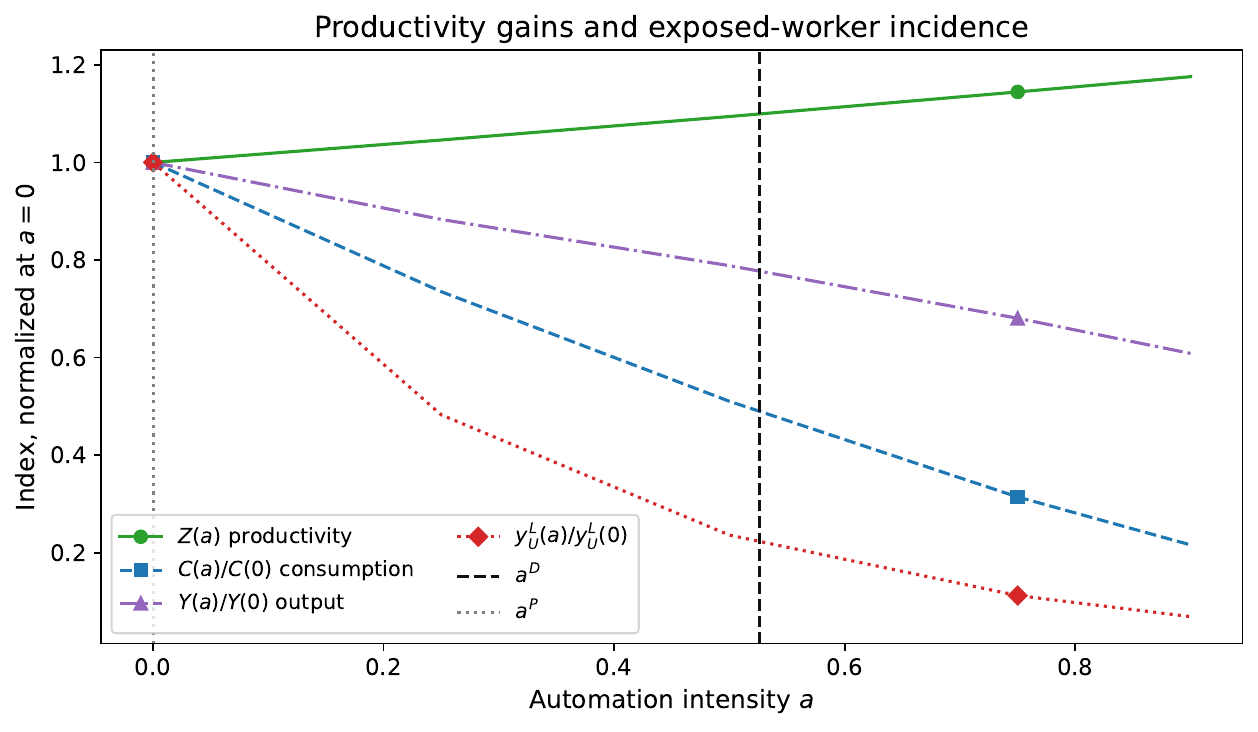}
\caption{Productivity and stationary-allocation responses. The figure reports the market-clearing stationary equilibrium associated with each automation value up to \(a=0.90\).}
\label{fig:tradeoff}
\end{figure}

\begin{figure}[H]
\centering
\includegraphics[width=0.82\textwidth]{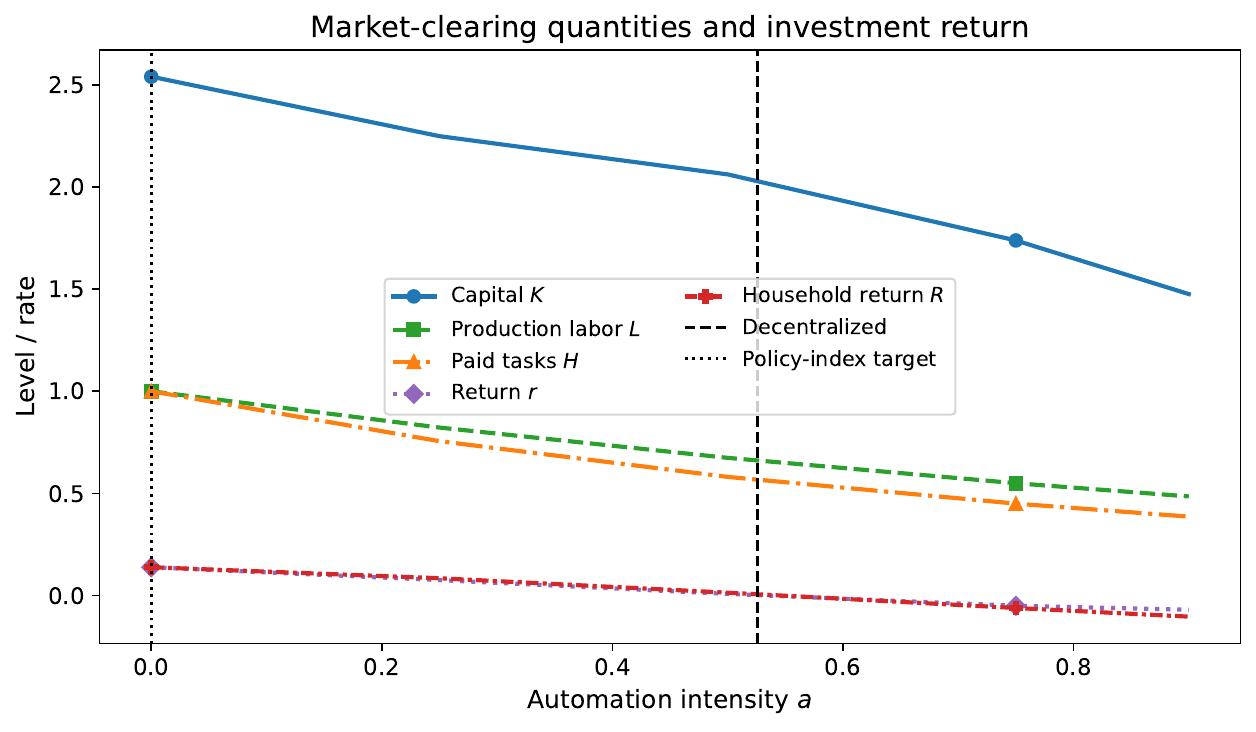}
\caption{Market-clearing quantities and investment return. The interest rate depends on production labor relative to capital and on automation-induced depreciation, not on capital scarcity alone.}
\label{fig:prices}
\end{figure}

Figures \ref{fig:tradeoff} and \ref{fig:prices} should be read together. Figure \ref{fig:tradeoff} shows rising productivity together with falling consumption and the collective wage bill in the adverse-incidence scenario. Figure \ref{fig:prices} shows the market-clearing adjustment. Capital falls at high \(a\), while production labor \(L\) also falls and depreciation \(\delta(a)\) rises. The investment return is therefore determined by the full marginal-product and resource-accounting system.

\subsection{Goods-market accounting}

A useful way to understand the numbers is to ask where output goes. The one-good accounting condition is
\[
Y=C+\delta(a)K+\Phi^A(a)+\omega_T\tau a+(1-\theta_E)\Pi^A.
\]
In the baseline no-tax economy, \(\tau=0\), so the fiscal term is zero. Table \ref{tab:resource_grid} and Figure \ref{fig:goods} report the accounting decomposition. As automation rises, consumption falls, replacement investment becomes a larger share of output because \(\delta(a)\) rises, and foreign rent leakage can become relevant when automation rents are not fully domestically owned. The final good is the numeraire, so the model does not literally say that consumption goods become cheaper. The real statement is that automation changes the feasible allocation of final-good output between consumption, replacement investment, automation costs, and rents.

\begin{table}[!htbp]
\centering
\caption{Resource accounting along the automation grid}
\label{tab:resource_grid}
\begin{tabular}{ccccccccc}
\toprule
$a$ & $Z$ & $K$ & $L$ & $Y$ & $C$ & $\delta(a)K$ & $y_U^L$ & $y_H^L$ \\
\midrule
0.00 & 1.000 & 2.540 & 1.000 & 1.399 & 1.246 & 0.152 & 0.671 & 1.119 \\
0.25 & 1.046 & 2.249 & 0.822 & 1.235 & 0.916 & 0.276 & 0.324 & 1.312 \\
0.50 & 1.094 & 2.062 & 0.674 & 1.103 & 0.636 & 0.381 & 0.159 & 1.560 \\
0.75 & 1.145 & 1.738 & 0.550 & 0.952 & 0.392 & 0.430 & 0.075 & 1.802 \\
0.90 & 1.176 & 1.476 & 0.485 & 0.851 & 0.269 & 0.421 & 0.046 & 1.887 \\
\bottomrule
\end{tabular}
\end{table}

\begin{figure}[H]
\centering
\includegraphics[width=0.86\textwidth]{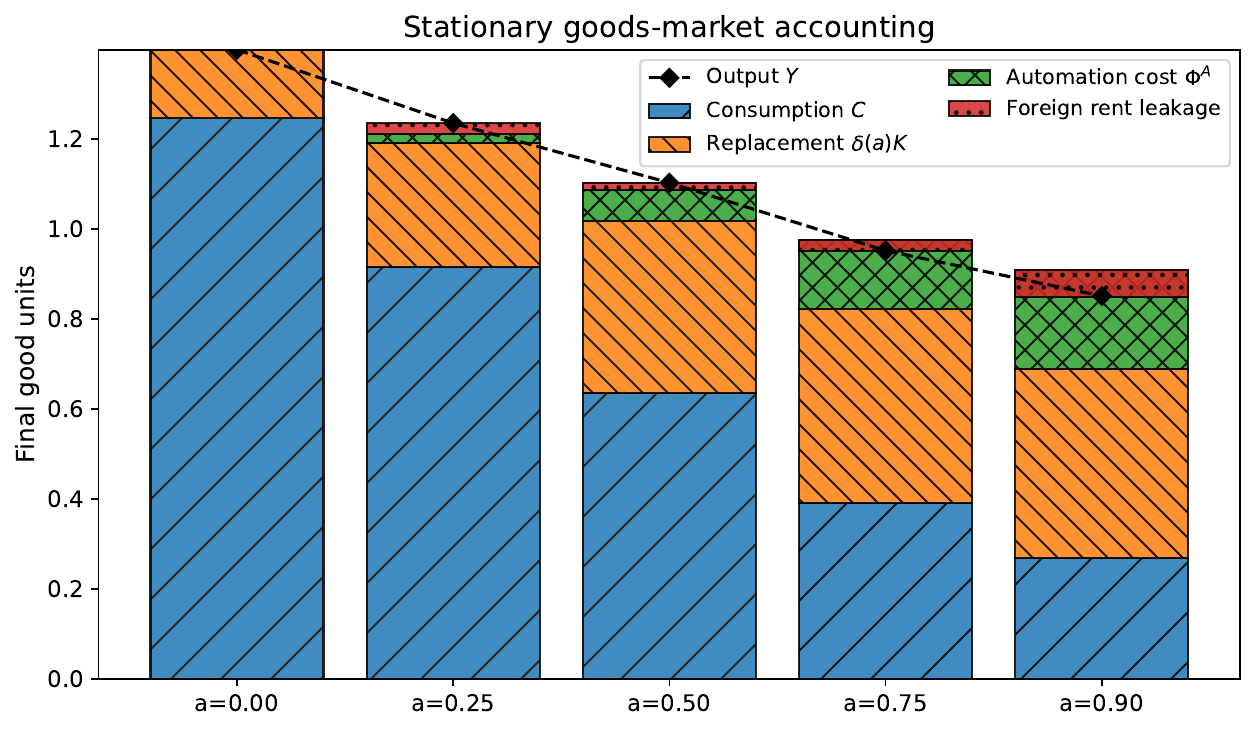}
\caption{Stationary goods-market accounting. Output is split into consumption, replacement investment, real automation costs, and foreign-owned automation rents. The plotted bars sum to output up to the numerical residual.}
\label{fig:goods}
\end{figure}

\subsection{Productivity-led capital growth}\label{subsec:capital_growth_case}

The baseline adverse-incidence case is not the only possible outcome. This subsection repeats the Section 4 analysis under a stronger productivity shifter, stronger high-skill complementarity, weaker low-skill exposure, and much lower automation-induced obsolescence. Table \ref{tab:capital_growth_params} lists the parameter changes. All other parameters, including the ownership benchmark \(\theta_E=0.45\), are held at their baseline values unless reported. This is the disciplined optimistic case: artificial intelligence raises productivity enough that households save more, capital grows, and aggregate consumption rises.

\begin{table}[!htbp]
\centering
\caption{Parameter changes for the productivity-led capital-growth case}
\label{tab:capital_growth_params}
\small
\setlength{\tabcolsep}{4pt}
\begin{tabular}{lp{0.38\textwidth}cc}
\toprule
Parameter & Meaning & Adverse-incidence baseline & Productivity-led case \\
\midrule
\(\psi_Z\) & Hicks-neutral productivity gain & 0.18 & 0.40 \\
\(\chi_U\) & low-skill paid-task exposure & 3.20 & 1.50 \\
\(\xi_U\) & low-skill production-task exposure & 2.50 & 0.70 \\
\(\beta_H\) & high-skill paid-task complementarity & 0.35 & 0.70 \\
\(\eta_H^L\) & high-skill production-task complementarity & 0.55 & 1.00 \\
\(\delta_A\) & automation-induced capital obsolescence & 0.25 & 0.02 \\
\(\kappa\) & convex automation cost & 0.52 & 3.00 \\
\bottomrule
\end{tabular}
\begin{minipage}{0.92\textwidth}
\vspace{0.35em}
\footnotesize
The productivity-led case is a counterfactual parameter regime that combines stronger productivity and complementarity with lower exposure and obsolescence.
\end{minipage}
\end{table}

Table \ref{tab:primitive_diagnostic_use} states how the diagnostic benchmark in Appendix \ref{app:diagnostic_benchmark} guides the two full-equilibrium numerical regimes.

\begin{table}[!htbp]
\centering
\caption{How the primitive diagnostic is used in the numerical regimes}
\label{tab:primitive_diagnostic_use}
\small
\begin{tabular}{p{0.24\textwidth}p{0.34\textwidth}p{0.34\textwidth}}
\toprule
Regime & Primitive forces in the diagnostic decomposition & Full stationary-GE check \\
\midrule
Adverse-incidence & Large low-skill exposure \(\chi_U,\xi_U\), large obsolescence \(\delta_A\), and the adverse mobility tilt make the negative low-skill term \(\Gamma_U\) dominate. & Solve the HJB--KFE fixed point; record the ordering of \(a^{D}\) and \(a^P\), changes in aggregate \(C\), and low-skill low-wealth consumption. \\
Productivity-led capital growth & Larger productivity gain \(\psi_Z\), high-skill complementarity \((\beta_H,\eta_H^L)\), and lower obsolescence weaken the negative exposed-household term and strengthen \(\Gamma_H\). & Solve the same fixed point; check whether \(K,C,Y\) rise and then separately test whether low-skill low-wealth households are protected. \\
\bottomrule
\end{tabular}
\end{table}

Table \ref{tab:no_wealth_diagnostic_values} reports the automation levels implied by the no-wealth-heterogeneity diagnostic benchmark and compares them with the stationary heterogeneous-agent computation. The diagnostic calculation sets \(\bar r=0.138\), equal to the no-automation market-clearing return in the full model, and uses \(C(a)=m_U(a)y_U^L(a)+m_H(a)y_H^L(a)\). Thus there is no wealth distribution or wealth-dependent MPC in this calculation. In the adverse-incidence case, the policy-index maximizer lies below the private root because the low-skill incidence term dominates. The full stationary GE computation yields the same ordering, \(a^{D}=0.526\) and \(a^P=0\). In the productivity-led case both diagnostic choices are positive, and the full model moves both values upward once capital accumulation, ownership, and market clearing are reintroduced.

\begin{table}[!htbp]
\centering
\caption{No-wealth diagnostic benchmark versus full stationary GE}
\label{tab:no_wealth_diagnostic_values}
\small
\begin{tabular}{lcccc}
\toprule
Regime & diagnostic $a^{D}$ & diagnostic $a^P$ & full GE $a^{D}$ & full GE $a^P$ \\
\midrule
Adverse-incidence baseline & 0.419 & 0.000 & 0.526 & 0.000 \\
Productivity-led capital growth & 0.376 & 0.201 & 0.427 & 0.375 \\
\bottomrule
\end{tabular}
\begin{minipage}{0.92\textwidth}
\vspace{0.35em}
\footnotesize
The diagnostic benchmark ignores asset heterogeneity, fiscal rebates, and automation-rent ownership. It solves the closed-form private residual \(F(a)=\mathcal M(a)-\phi-\kappa a\) and maximizes the reduced-form index \(\lambda C(a)+\mu B_U(a)\). The full GE columns are the computational results from the stationary market-clearing HJB--KFE economy.
\end{minipage}
\end{table}

Table \ref{tab:capital_growth_case} reports the stationary allocation at \(a=0\) and the decentralized automation allocation under the productivity-led parameter vector. In this case, the decentralized firm chooses positive automation, output rises, aggregate consumption rises, and the stationary capital stock rises. Capital deepening lowers the investment return \(r\). The household asset return \(R\) includes domestic automation-rent ownership and may differ from \(r\).

\begin{table}[!htbp]
\centering
\caption{Productivity-led capital-growth case: stationary aggregates}
\label{tab:capital_growth_case}
\scriptsize
\setlength{\tabcolsep}{3pt}
\begin{tabular}{lcccccccccc}
\toprule
Allocation & $a$ & $K$ & $L$ & $H$ & $Y$ & $C$ & $B$ & $r$ & $R$ & $w$ \\
\midrule
No automation & 0.000 & 2.540 & 1.000 & 1.000 & 1.399 & 1.246 & 0.895 & 0.138 & 0.138 & 0.895 \\
Decentralized automation & 0.427 & 4.144 & 1.025 & 0.841 & 2.011 & 1.474 & 1.055 & 0.106 & 0.101 & 1.255 \\
\bottomrule
\end{tabular}
\begin{minipage}{0.92\textwidth}
\vspace{0.35em}
\footnotesize
Here \(B=wH\) is the collective wage bill. The fall in \(r\) reflects capital deepening. The lower value of \(R\) relative to \(r\) occurs because net automation rents \(\Pi^A\) are negative in this particular capital-growth counterfactual after real automation costs are netted out.
\end{minipage}
\end{table}

The household budget equation explains how low-capital households fund consumption in any of the stationary experiments:
\[
\dot k_s(k)=Rk+w e_sh_s(a)+T_s(k)-c_s(k).
\]
For a low-capital household, \(Rk\) is small, so consumption is financed mainly by labor income and transfers, with the borrowing constraint preventing persistent negative drift at the lower boundary. Group-level income tables should not be read as closed budgets for each skill group, because households move between skill states carrying their assets with them. The aggregate household budget and the goods-market constraint close in the whole economy, not separately within each skill group.

Figures \ref{fig:capital_growth_agg_returns}--\ref{fig:capital_growth_consumption} repeat the same analysis used for the baseline case. Figure \ref{fig:capital_growth_agg_returns} groups the aggregate and asset-market information together: panel (a) plots
\[
K^{D}/K^0,\qquad C^{D}/C^0,\qquad Y^{D}/Y^0,\qquad B^{D}/B^0,
\]
where superscript \(0\) denotes the no-automation allocation under the same parameter vector; panel (b) plots the no-automation and decentralized values of
\[
r=\alpha Z(a)K^{\alpha-1}L^{1-\alpha}-\delta(a),\qquad
R=r+\theta_E\Pi^A/K,
\]
and the dividend yield \(R-r\). Figure \ref{fig:capital_growth_group} decomposes group income sources and consumption. The first two bars in Figure \ref{fig:capital_growth_group} are the no-automation allocation and the last two bars are the decentralized automation allocation. Figures \ref{fig:capital_growth_density} and \ref{fig:capital_growth_consumption} report wealth densities and consumption functions.

\begin{figure}[H]
\centering
\includegraphics[width=0.96\textwidth]{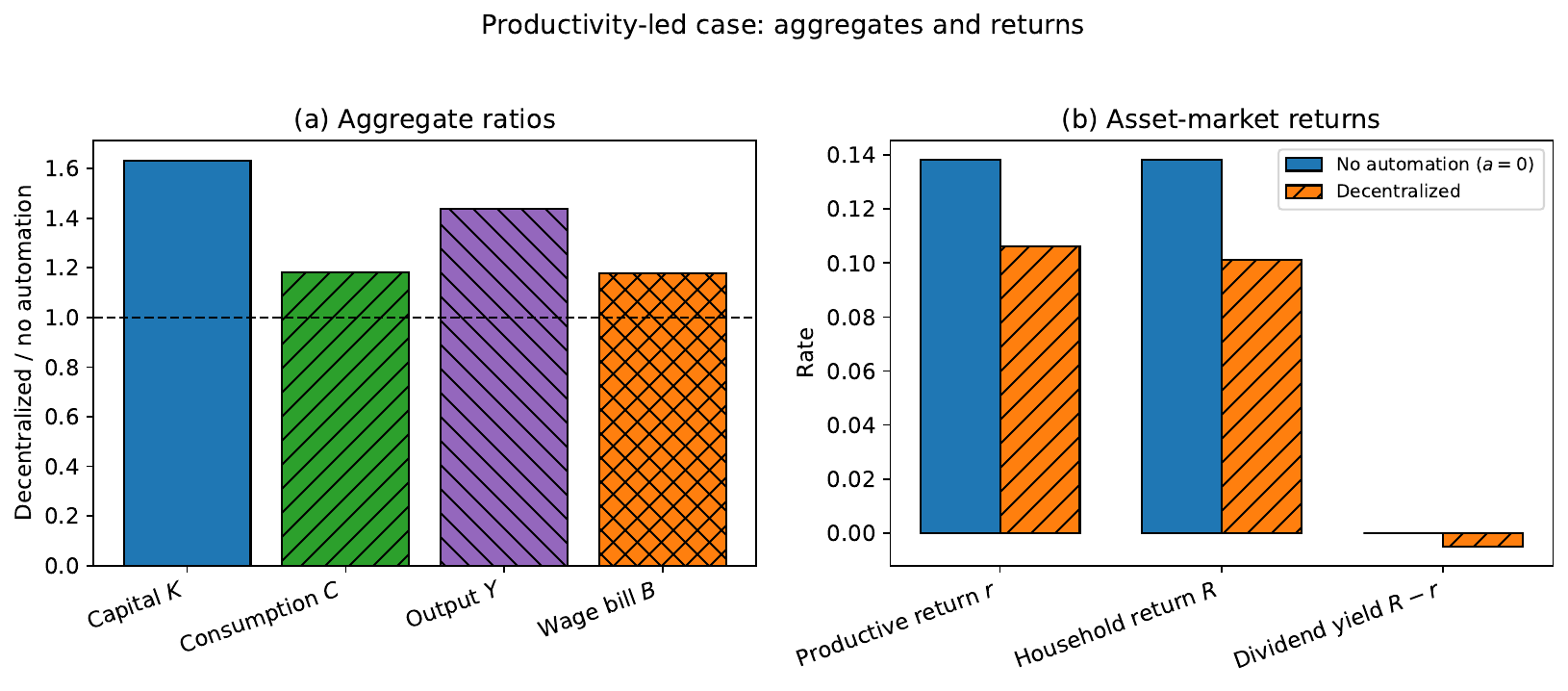}
\caption{Productivity-led capital-growth case: aggregate ratios and returns. Panel (a) shows that decentralized automation raises capital, consumption, output, and the collective wage bill relative to no automation. Panel (b) shows that the productive return \(r\) and household return \(R\) can fall when capital becomes abundant, even though aggregate output and consumption rise. }
\label{fig:capital_growth_agg_returns}
\end{figure}

\begin{figure}[H]
\centering
\includegraphics[width=0.90\textwidth]{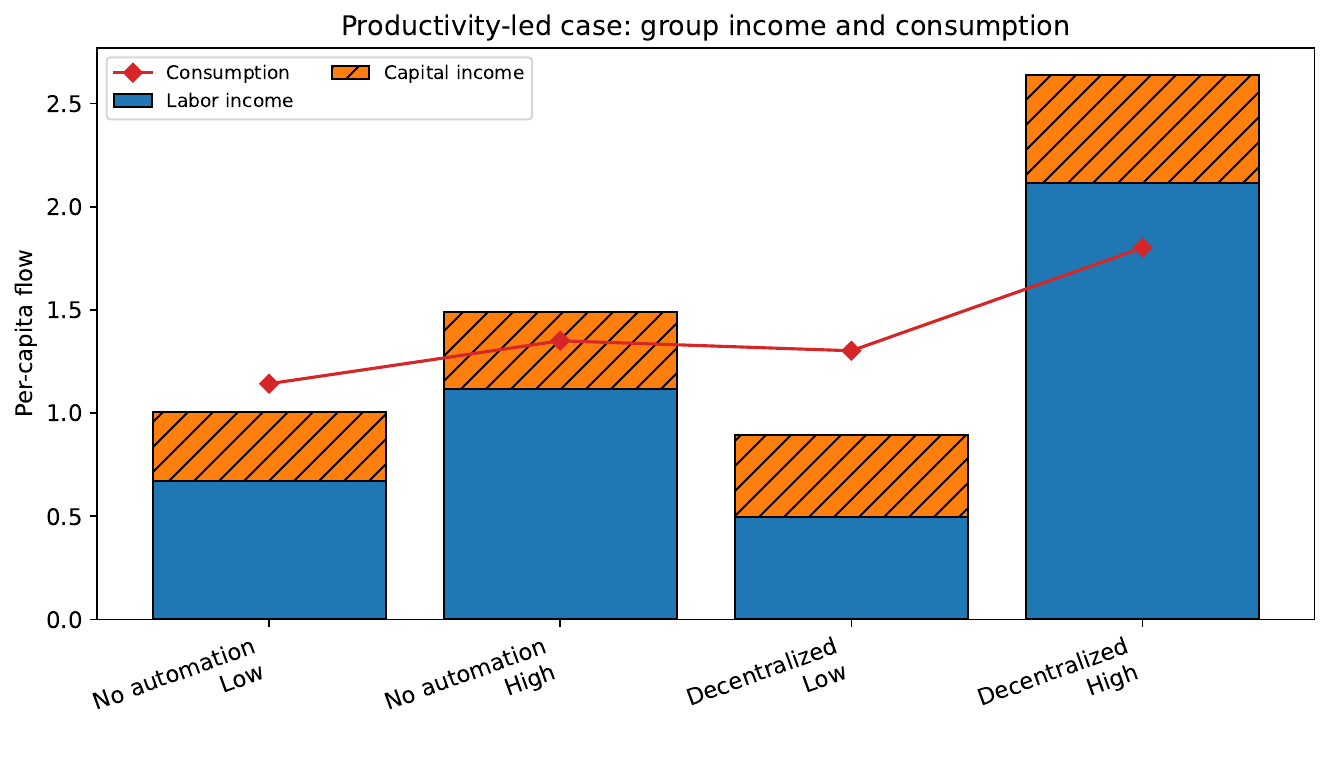}
\caption{Productivity-led capital-growth case: group income and consumption. The first two bars are the no-automation allocation and the last two bars are the decentralized automation allocation. Bars show per-capita labor income \(we_sh_s(a)\) and capital income \(Rk\); markers show mean consumption. The figure explains why the optimistic regime can raise aggregate consumption even though low-skilled task income is still reduced. Tiny automation-dividend components are omitted from the plot when they are visually negligible.}
\label{fig:capital_growth_group}
\end{figure}

Aggregate abundance does not automatically imply distributional protection of exposed poor households. The relevant object is not only aggregate consumption \(C\), but consumption conditional on skill and low wealth. Define the low-wealth cutoff \(\bar k_L=1.20\), the bottom part of the no-automation wealth distribution in this calibration, and compute
\[
\bar c_{s,\leq \bar k_L}^P
=\frac{\int_{\underline k}^{\bar k_L}c_s^P(k)g_s^P(k)\,dk}
{\int_{\underline k}^{\bar k_L}g_s^P(k)\,dk},
\qquad P\in\{0,D\}.
\]
Table \ref{tab:capital_growth_low_wealth} reports this conditional consumption statistic. The productivity-led case raises aggregate consumption, but low-skill low-wealth consumption falls relative to the no-automation allocation. This is the sense in which the optimistic case establishes an aggregate-abundance possibility, while distributional protection of poor exposed workers still requires broad pass-through, ownership, or tax-and-rebate transfers. Table \ref{tab:capital_growth_low_wealth} reports the conditional means, and Figure \ref{fig:capital_growth_low_wealth} plots the same comparison.

\begin{table}[H]
\centering
\caption{Productivity-led case: low-wealth conditional consumption}
\label{tab:capital_growth_low_wealth}
\begin{tabular}{lccccc}
\toprule
Allocation & Group & Cell mass & Mean \(k\) & Labor income & Mean consumption \\
\midrule
No automation & Low skill & 0.180 & 0.649 & 0.671 & 0.844 \\
No automation & High skill & 0.148 & 0.729 & 1.119 & 1.051 \\
Decentralized & Low skill & 0.161 & 0.590 & 0.496 & 0.702 \\
Decentralized & High skill & 0.040 & 0.757 & 2.115 & 1.129 \\
\bottomrule
\end{tabular}
\begin{minipage}{0.90\textwidth}
\vspace{0.35em}
\footnotesize
The table conditions on households with \(k\leq1.20\). The key row is low skill under decentralized automation: aggregate \(C\) rises in the productivity-led regime, but consumption of low-skill low-wealth households falls because they have little asset income and their labor income remains exposed.
\end{minipage}
\end{table}

\begin{figure}[H]
\centering
\includegraphics[width=0.82\textwidth]{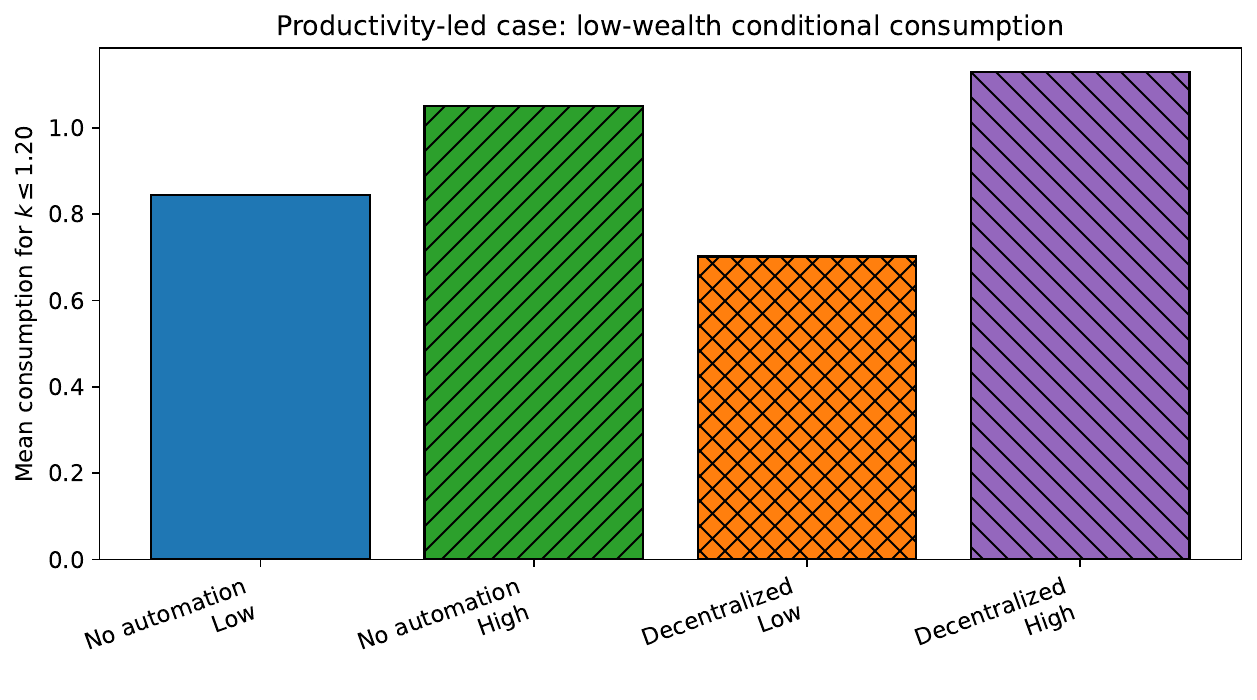}
\caption{Productivity-led case: conditional consumption among low-wealth households. The plotted statistic is \(\bar c_{s,\leq \bar k_L}^P\) with \(\bar k_L=1.20\). The figure makes explicit that aggregate abundance need not protect low-skill low-wealth households absent broad ownership or transfers.}
\label{fig:capital_growth_low_wealth}
\end{figure}

\begin{figure}[H]
\centering
\includegraphics[width=0.82\textwidth]{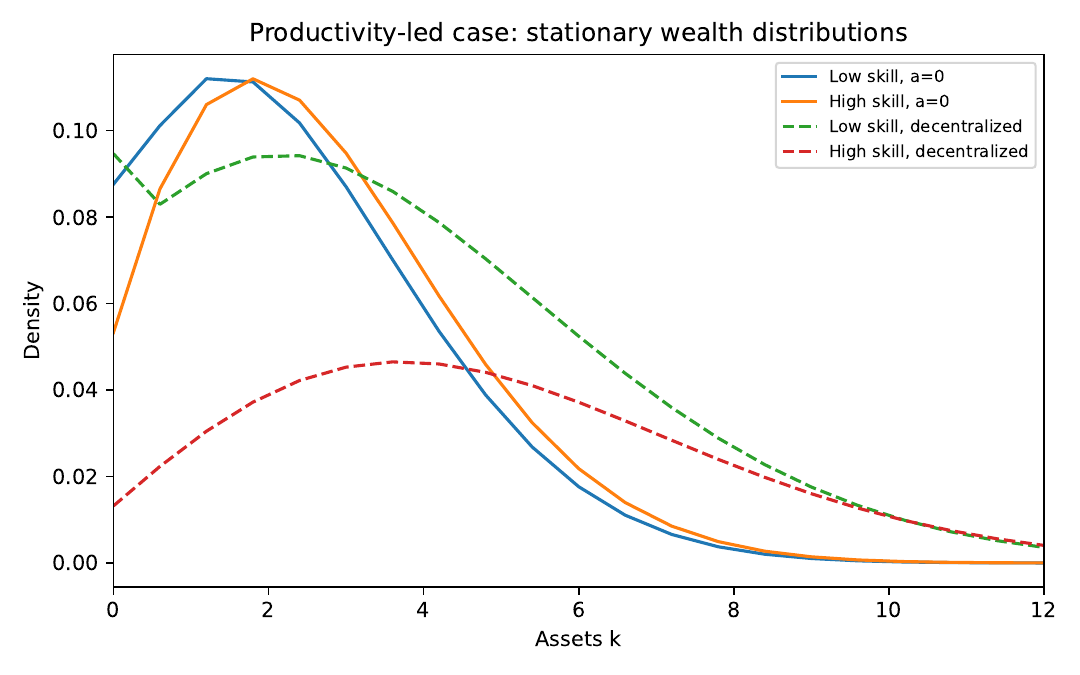}
\caption{Productivity-led capital-growth case: stationary wealth distributions \(g_s(k)\). Both the no-automation and decentralized distributions are plotted. Higher aggregate saving shifts the distribution toward larger asset holdings, but the low-wealth region near the borrowing constraint remains relevant for distributional incidence.}
\label{fig:capital_growth_density}
\end{figure}

\begin{figure}[H]
\centering
\includegraphics[width=0.82\textwidth]{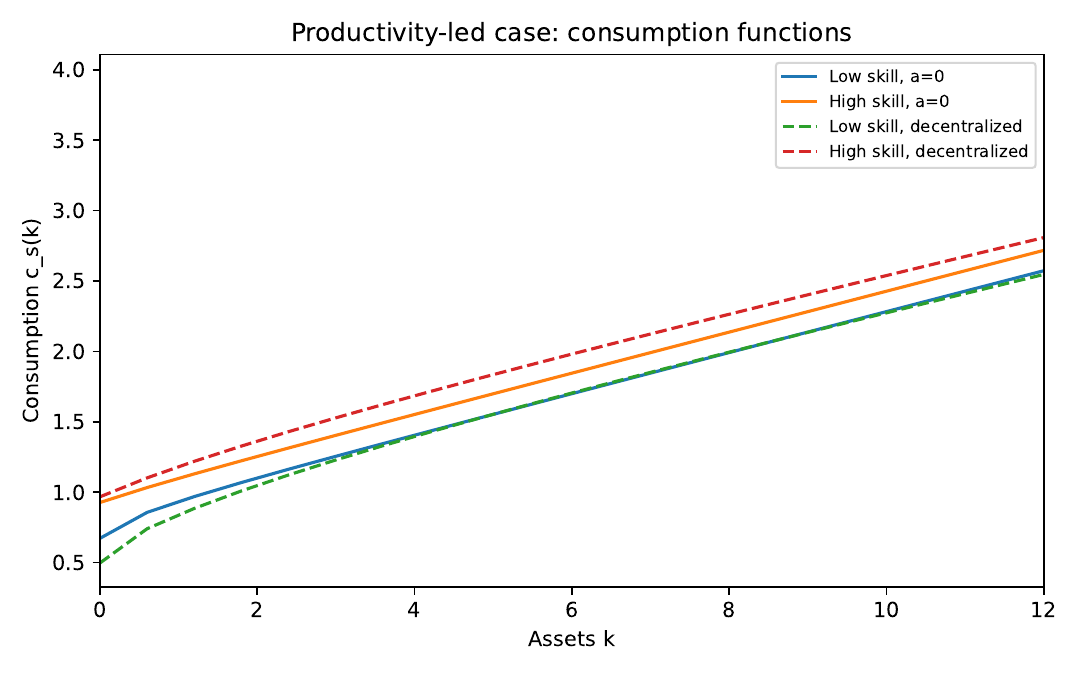}
\caption{Productivity-led capital-growth case: consumption functions \(c_s(k)\). Solid lines are the no-automation allocation and dashed lines are the decentralized automation allocation. Read this figure together with Figure \ref{fig:capital_growth_density}: consumption at a given \(k\) is only half of the distributional question; the density \(g_s(k)\) tells where households actually are.}
\label{fig:capital_growth_consumption}
\end{figure}
\FloatBarrier

A tax has a different stationary incidence in the productivity-led regime than in the adverse-incidence regime. In the capital-growth case, moving from decentralized automation to the no-automation allocation reduces \(C\) from 1.474 to 1.246 and reduces \(K\) from 4.144 to 2.540. Thus a tax that eliminates automation supports exposed labor income but sacrifices productivity growth and capital accumulation in this scenario. Positive-revenue taxes with targeted transfers illustrate the distributional tradeoff without being labeled optimal policy.

\subsection{Distributional incidence and consumption smoothing}

The key distributional state is \((s,k)\), not skill alone. A low-skilled household with high wealth owns capital and may be partly protected. A high-skilled household with little wealth still depends on labor income. The local marginal propensity to consume out of wealth is defined in \eqref{eq:wealth_mpc}. This wealth-MPC is computed from the finite-difference consumption policy and is used mainly in the fiscal-incidence exercise, especially Table \ref{tab:progressive_rebate} and Figure \ref{fig:progressive_targeted_diagnostics}. It is an incidence statistic that summarizes the immediate consumption support from transfers. The consumption-equivalent gain from policy \(P\) relative to the decentralized allocation is the constant proportional increase in consumption in the decentralized allocation that would make the household indifferent between the two allocations. It is defined by
\[
V_s^P(k)=(1+CE_s(k))^{1-\gamma}V_s^{D}(k),
\]
so that, for \(\gamma\neq1\),
\[
CE_s(k)=\left(\frac{V_s^P(k)}{V_s^{D}(k)}\right)^{1/(1-\gamma)}-1.
\]
Group-level CE uses the decentralized distribution as weights. Table \ref{tab:incidence} reports these distributional effects for the adverse-incidence baseline.

\begin{table}[!htbp]
\centering
\caption{Distributional incidence of the implementing policy}
\label{tab:incidence}
\begin{tabular}{lrrrrrrrr}
\toprule
Group & Mass & Mean $k$ & $y^{L,D}$ & $y^{L,P}$ & $c^{D}$ & $c^{P}$ & $c^P/c^{D}-1$ & CE \\
\midrule
Low bottom & 0.353 & 0.56 & 0.147 & 0.671 & 0.287 & 0.822 & 1.86 & 1.98 \\
Low middle & 0.169 & 2.07 & 0.147 & 0.671 & 0.568 & 1.111 & 0.96 & 1.49 \\
Low top & 0.165 & 3.89 & 0.147 & 0.671 & 0.825 & 1.387 & 0.68 & 1.36 \\
High bottom & 0.083 & 0.75 & 1.586 & 1.119 & 0.595 & 1.055 & 0.77 & 1.32 \\
High middle & 0.084 & 2.10 & 1.586 & 1.119 & 0.811 & 1.268 & 0.56 & 1.18 \\
High top & 0.145 & 4.17 & 1.586 & 1.119 & 1.084 & 1.576 & 0.45 & 1.15 \\
\bottomrule
\end{tabular}
\begin{minipage}{0.93\textwidth}
\vspace{0.35em}
\footnotesize
The table reports comparisons between stationary environments, weighted by the decentralized distribution. The largest CE differences occur for low-skill bottom-wealth households. These are not transition-path welfare effects.
\end{minipage}
\end{table}

Figure \ref{fig:density} plots the stationary wealth density \(g_s(k)\) separately for each skill group and policy regime. The solid curves are the decentralized distribution; the dashed curves correspond to implementation of the policy-index target. Figure \ref{fig:consumption} then reports consumption conditional on wealth, and Figure \ref{fig:consumption_gain} plots the difference. Reading the three figures together shows both incidence and smoothing: low-wealth households have the largest local consumption response, while high-wealth households smooth with asset income.

\begin{figure}[!htbp]
\centering
\includegraphics[width=0.82\textwidth]{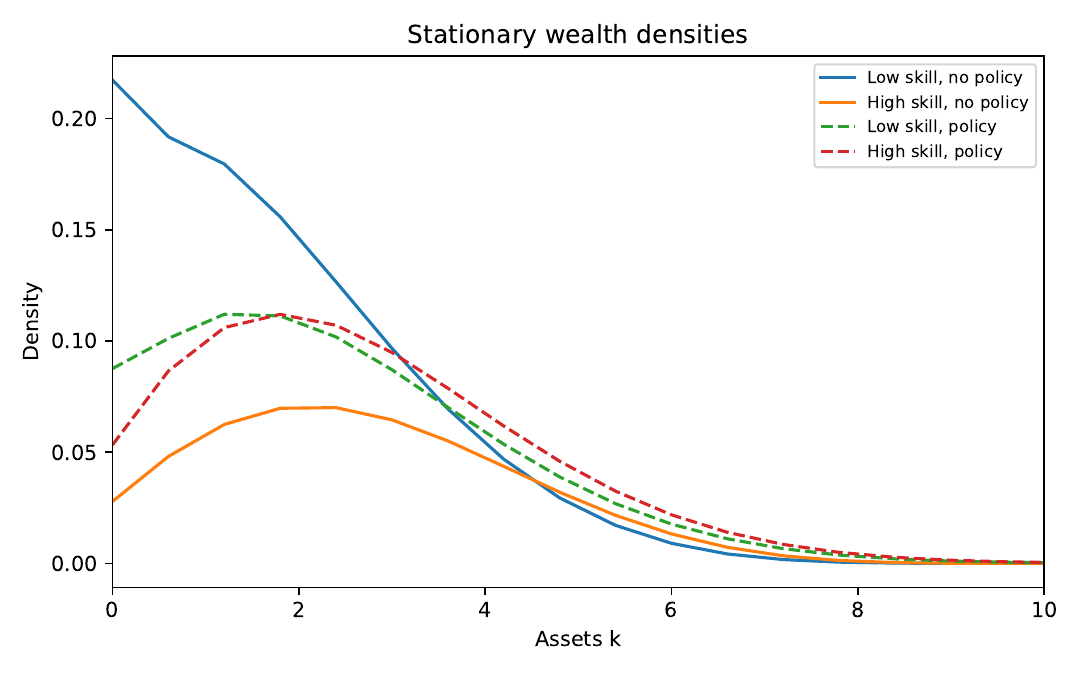}
\caption{Stationary wealth densities with and without the implementing policy. The policy changes not only current income but also the long-run distribution of wealth.}
\label{fig:density}
\end{figure}

\begin{figure}[!htbp]
\centering
\includegraphics[width=0.82\textwidth]{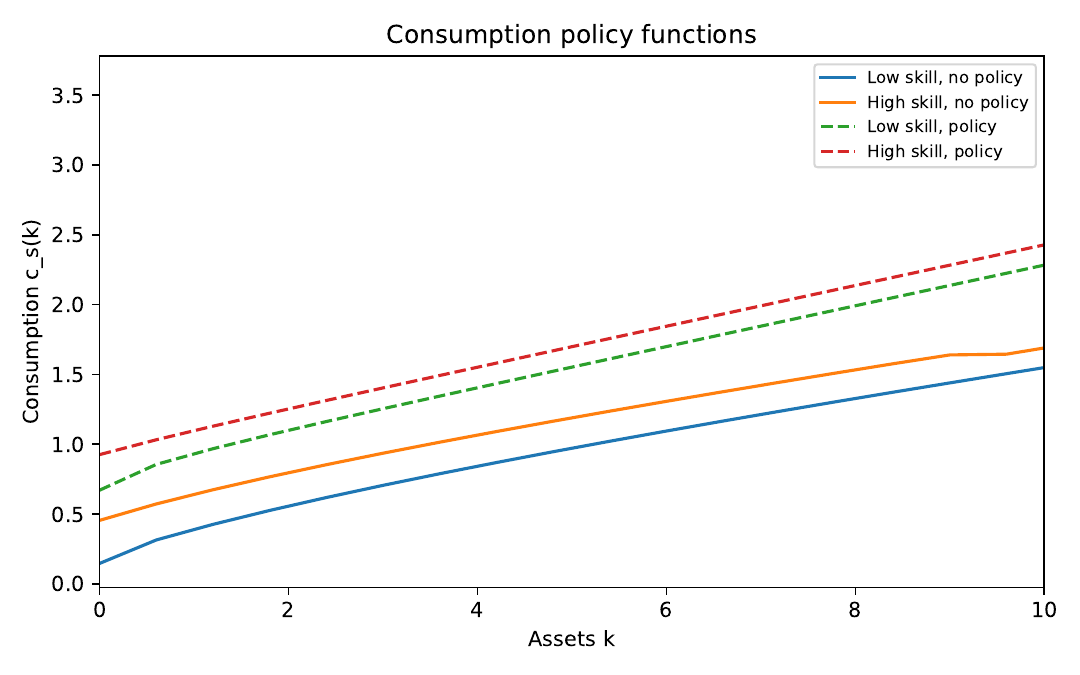}
\caption{Consumption policy functions with and without the policy. This is the numerical evidence on consumption smoothing: consumption is increasing in wealth, and low-wealth households have higher local MPCs.}
\label{fig:consumption}
\end{figure}

\begin{figure}[!htbp]
\centering
\includegraphics[width=0.82\textwidth]{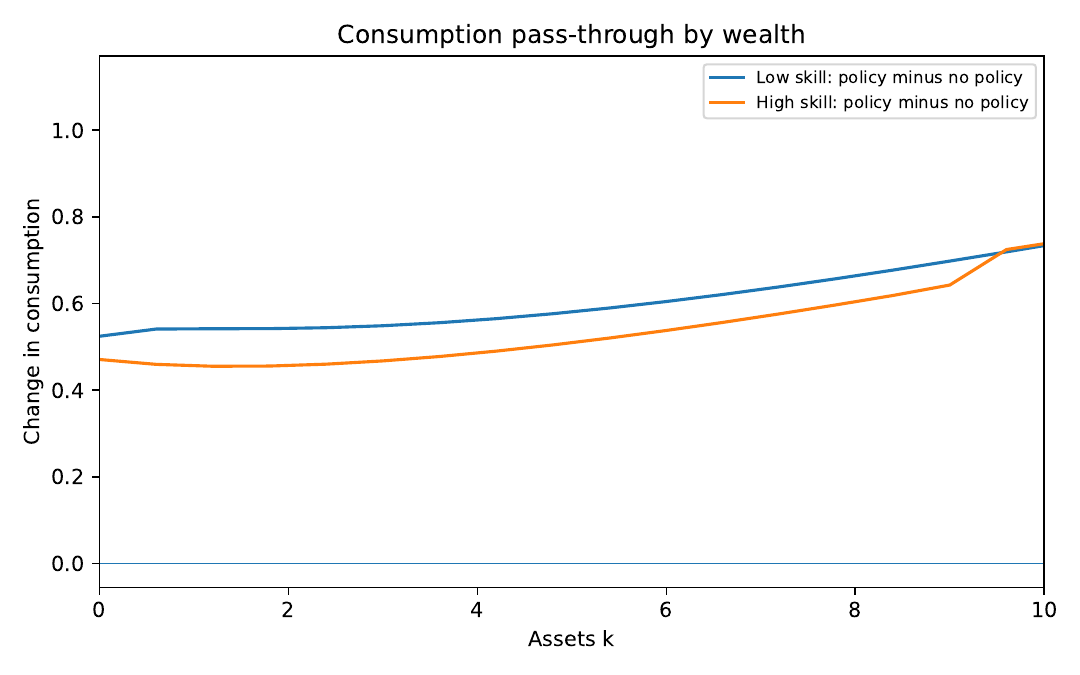}
\caption{Policy-induced consumption change as a function of wealth. The plotted object is \(c_s^{P}(k)-c_s^{D}(k)\); the horizontal zero line denotes no consumption change. Low-wealth households have the largest local consumption response.}
\label{fig:consumption_gain}
\end{figure}

\subsection{Static benchmark versus stationary GE}

Section 2 separates the firm's direct profit derivative from the derivative of the policy index. The no-wealth benchmark in Appendix \ref{app:diagnostic_benchmark} expresses their ordering in primitive exposure and productivity terms, while the stationary GE block recomputes prices and the wealth distribution. In the diagnostic benchmark, the private firm pushes automation to the upper grid value. In the stationary economy, capital-market clearing, depreciation, and the wealth distribution moderate the private choice to \(a^{D}=0.526\). The reduced-form index is maximized at \(a^P=0\) in both calculations. The heterogeneous-agent model additionally reports exposure by wealth, capital adjustment, and consumption smoothing. Table \ref{tab:static_vs_stationary} summarizes the comparison.

\begin{table}[!htbp]
\centering
\caption{Static benchmark versus stationary heterogeneous-agent GE, \(\lambda=0.6\)}
\label{tab:static_vs_stationary}
\small
\setlength{\tabcolsep}{3pt}
\begin{tabular}{lcccccccc}
\toprule
Model & private $a$ & index $a^P$ & $K^{D}$ & $K^{P}$ & $C^{D}$ & $C^{P}$ & $B^{D}$ & $B^{P}$ \\
\midrule
Static benchmark & 0.900 & 0.000 & -- & -- & 0.903 & 0.895 & 0.903 & 0.895 \\
Stationary Aiyagari/Huggett GE & 0.526 & 0.000 & 2.036 & 2.540 & 0.609 & 1.246 & 0.597 & 0.895 \\
\bottomrule
\end{tabular}
\end{table}

\subsection{Synthesis: two stationary scenario families}

Table \ref{tab:regime_map} summarizes the parameter forces behind the two scenario families, and Table \ref{tab:computed_regimes} reports the corresponding computed examples. The productivity-led case arises when the productivity gain \(\psi_Z\), high-skill complementarity \((h_H',\ell_H')\), and domestic ownership \(\theta_E\) are large. The adverse-incidence case combines strong exposure, adverse mobility, obsolescence, and concentrated ownership. These are stationary comparative statics; the model does not map either case into a short-run forecast.

\begin{table}[!htbp]
\centering
\caption{Parameter regimes that separate the two views of AI}
\label{tab:regime_map}
\small
\begin{tabular}{p{0.24\textwidth}p{0.33\textwidth}p{0.33\textwidth}}
\toprule
Object & Productivity-led AI & Adverse-incidence \\
\midrule
Productivity and high-skill channel & High \(\psi_Z\), high \(h_H'(a)\), high \(\ell_H'(a)\). High-skilled workers gain and productivity gains are broad. & Moderate \(\psi_Z\). High-skilled workers gain, but low-skilled income losses are larger. \\
Exposure and mobility & Low or moderate \(\chi_U\); neutral or upward mobility; exposed households have assets or transfers. & High \(\chi_U\); adverse mobility; exposed households are low wealth and high-MPC. \\
Ownership & High domestic ownership \(\theta_E\), broad participation in capital and automation rents. & Low or concentrated ownership; equity income mainly insulates wealthy households. \\
Stationary outcome & Automation raises output, consumption, and capital. & Automation raises high-skill income but lowers consumption and capital in the reported stress test. \\
\bottomrule
\end{tabular}
\end{table}

\begin{table}[!htbp]
\centering
\caption{Computed examples of the two AI regimes}
\label{tab:computed_regimes}
\small
\begin{tabular}{lccccccc}
\toprule
Regime & \(\psi_Z\) & \(\delta_A\) & \(a^{D}\) & \(K^{D}/K^0\) & \(C^{D}/C^0\) & \(Y^{D}/Y^0\) & \(y_H^{D}/y_H^0\) \\
\midrule
Adverse-incidence baseline & 0.18 & 0.25 & 0.526 & 0.801 & 0.488 & 0.778 & 1.417 \\
Productivity-led capital-growth case & 0.40 & 0.02 & 0.427 & 1.631 & 1.183 & 1.438 & 1.890 \\
\bottomrule
\end{tabular}
\end{table}

Figure \ref{fig:two_regimes} compares the two computed regimes by plotting the ratios \(K^{D}/K^0\), \(C^{D}/C^0\), \(Y^{D}/Y^0\), and \(B^{D}/B^0\). Figure \ref{fig:scenario_map} then reports the parameter values behind the two scenario calculations.

\begin{figure}[!htbp]
\centering
\includegraphics[width=0.88\textwidth]{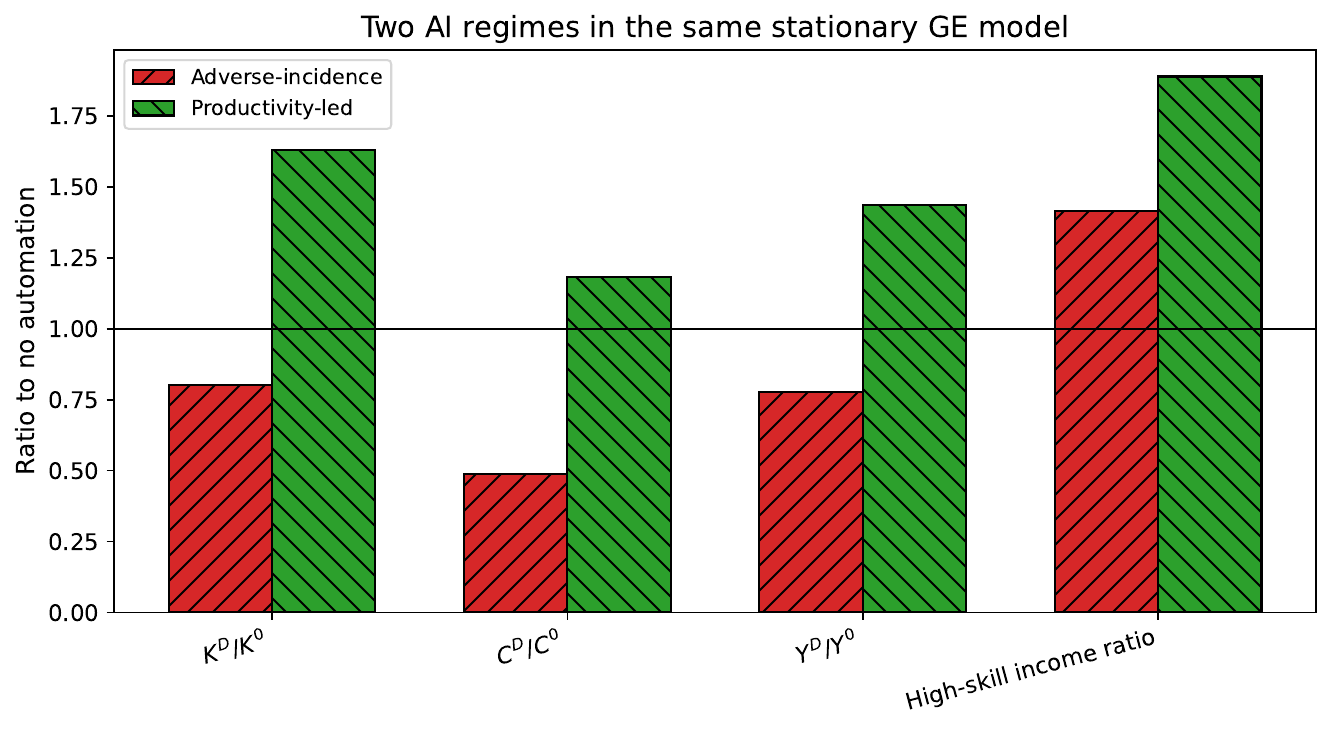}
\caption{Two AI regimes in the same stationary GE model. Bars above one mean that decentralized AI adoption raises the object relative to the no-automation allocation under the same parameter vector.}
\label{fig:two_regimes}
\end{figure}

\begin{figure}[!htbp]
\centering
\includegraphics[width=0.86\textwidth]{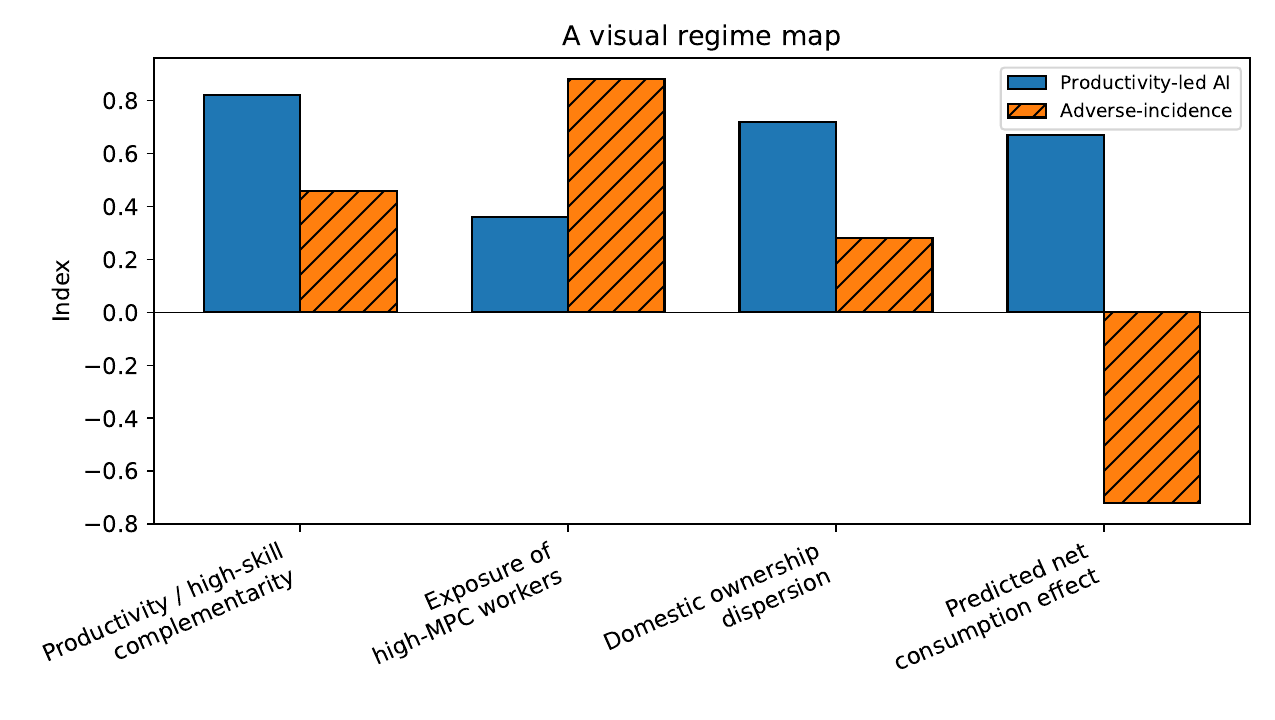}
\caption{Parameter values behind the two AI regimes. The horizontal axis gives the primitive parameter symbols and the corresponding values for the productivity-led and adverse-incidence cases. Bars report the numerical values used in the scenario calculations; colors, markers, and hatches distinguish regimes in both color and black-and-white print.}
\label{fig:scenario_map}
\end{figure}

\section{Taxes, Rebate Rules, and Ownership Incidence}

The implementing object is the marginal derivative that enters the firm automation condition \eqref{eq:firm_automation_foc}. For a literal automation tax \(P(a)=\tau a\), this derivative is
\[
P_a(a)=\tau,
\]
so the tax shifts the firm's first-order condition by \(\tau\). But a tax is also a fiscal object. If automation remains positive, the government collects \(\tau a\); if tax collection is costly, \(\omega_T\tau a\) is dissipated; and the remaining amount is rebated through \(T_s(k)\) as in \eqref{eq:rebate_schedule}. Unless otherwise stated, the fiscal experiments use \(\theta_E=0.45\), \(\lambda=0.6\), \(\mu=1\), and \(\omega_T=0.15\). Table \ref{tab:tax_friction} reports interior tax experiments. These are not the boundary implementation; they are designed to show fiscal incidence when revenue is actually collected.

The reported consumption-equivalent (CE) statistics compare household value functions across stationary environments and aggregate them using the decentralized stationary distribution. They are not welfare measures along a transition path from one invariant distribution to another. A dynamic policy evaluation would require the transition of prices and the wealth distribution, which is outside the model solved here.

The rebate looks small in the figures because it is a flow tax on the remaining automation base. A large marginal tax can change firm behavior a lot by changing the first-order condition, but the revenue it collects is \(\tau a(\tau)\). If the tax pushes automation close to zero, the tax base disappears. In the boundary implementation \(a=0\), revenue and rebates are exactly zero. In the interior example \(\tau=0.10\), revenue is only \(0.0393\), and the friction-adjusted lump-sum rebate is \(0.0334\) per person. That is visible when plotted separately, but it is small relative to high-skilled labor income and total consumption.

\begin{table}[!htbp]
\centering
\caption{Tax and fiscal-friction sensitivity, lump-sum rebates}
\label{tab:tax_friction}
\begin{tabular}{cccccccc}
\toprule
$\tau$ & $\omega_T$ & $a$ & $C$ & revenue & rebate & lost & avg. CE \\
\midrule
0.10 & 0.15 & 0.393 & 0.763 & 0.0393 & 0.0334 & 0.0059 & 0.406 \\
0.20 & 0.15 & 0.284 & 0.891 & 0.0567 & 0.0482 & 0.0085 & 0.735 \\
0.589 & 0.15 & 0.000 & 1.246 & 0.0000 & 0.0000 & 0.0000 & 1.562 \\
\bottomrule
\end{tabular}
\end{table}

Figures \ref{fig:tax_auto}--\ref{fig:tax_loss} separate the three fiscal objects. Figure \ref{fig:tax_auto} shows the behavioral response \(a^{D}(\tau)\). Figure \ref{fig:tax_ce} reports the stationary CE comparison. Figure \ref{fig:tax_loss} reports tax revenue and dissipated resources. This separation is important because a tax can have a large marginal effect on firm behavior while raising little revenue when the automation base shrinks.

\begin{figure}[!htbp]
\centering
\includegraphics[width=0.82\textwidth]{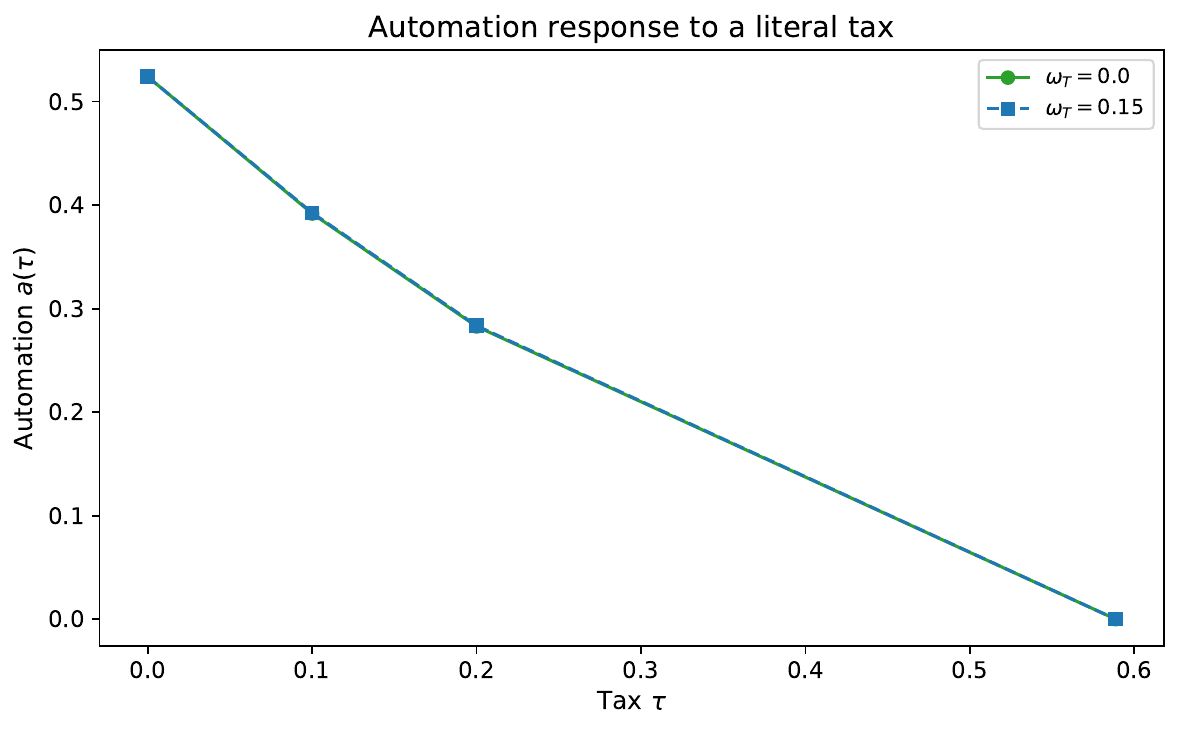}
\caption{Automation response to a literal automation tax. The tax reduces the automation base \(a(\tau)\), so the revenue base also shrinks.}
\label{fig:tax_auto}
\end{figure}

\begin{figure}[!htbp]
\centering
\includegraphics[width=0.82\textwidth]{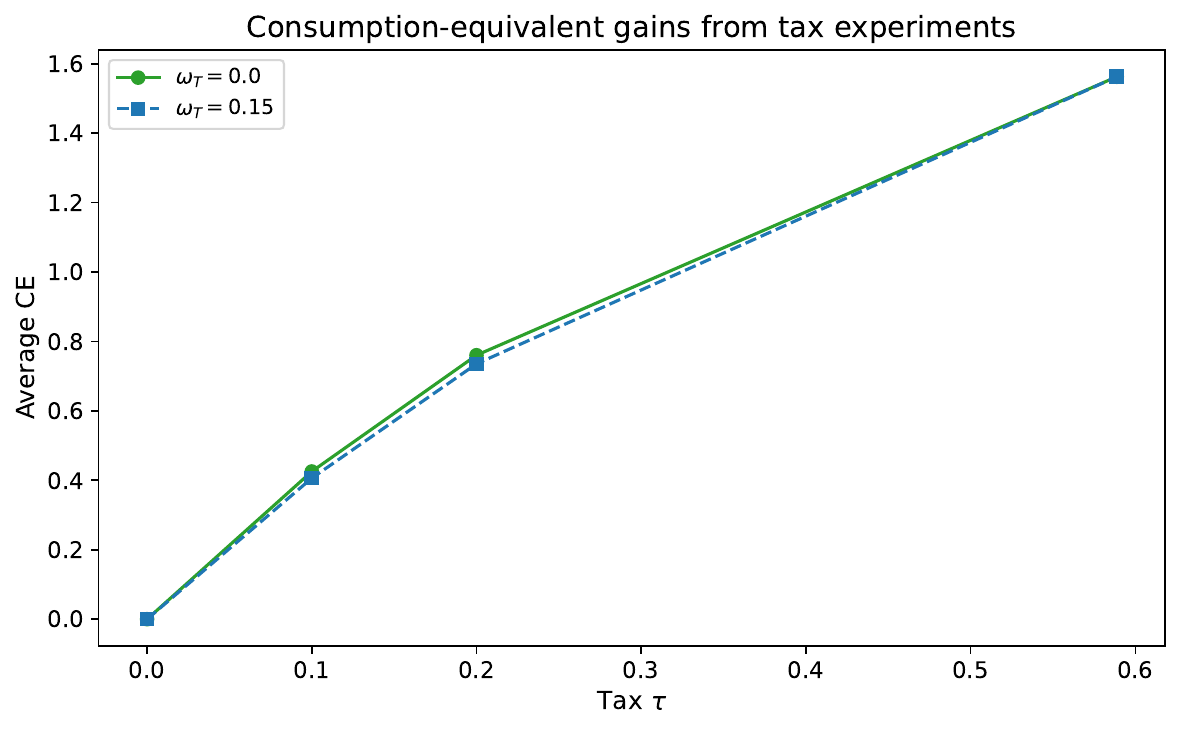}
\caption{Stationary consumption-equivalent comparisons for tax experiments. These compare value functions under invariant environments; they do not include transition-path welfare.}
\label{fig:tax_ce}
\end{figure}

\begin{figure}[!htbp]
\centering
\includegraphics[width=0.86\textwidth]{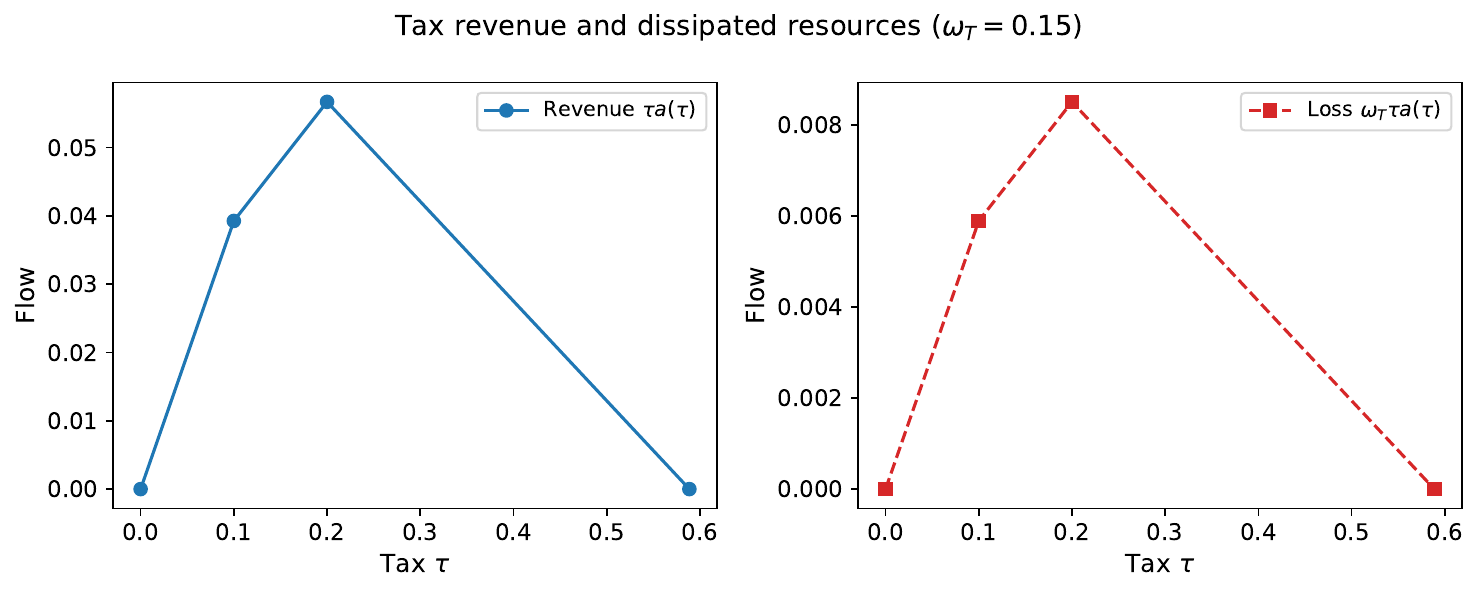}
\caption{Tax revenue and dissipated resources. A stationary flow is a per-period flow measured in units of the final good. The boundary tax deters automation and therefore collects zero revenue in the stationary allocation.}
\label{fig:tax_loss}
\end{figure}

Rebate rules change incidence. A lump-sum rebate returns the same flow to every household. A labor-income-proportional rebate gives more to high labor-income households. A total-income-proportional rebate also favors high-wealth households. Table \ref{tab:rebate_rules} compares rebate rules. The same marginal tax has a similar automation effect but different distributional effects. The final row reports the pro-consumption progressive case with \(\tau=0.20\); it is included in the same table to compare the aggregate effect of targeting rebates to low-wealth and low-earnings households.

\begin{table}[!htbp]
\centering
\caption{Rebate-rule comparison and pro-consumption rebate case}
\label{tab:rebate_rules}
\begin{tabular}{lcccccc}
\toprule
Rule & $a$ & $K$ & $C$ & rebate & lost & avg. CE \\
\midrule
Lump-sum & 0.393 & 1.991 & 0.763 & 0.0334 & 0.0059 & 0.406 \\
Labor-income proportional & 0.400 & 2.102 & 0.759 & 0.0340 & 0.0060 & 0.343 \\
Total-income proportional & 0.401 & 2.115 & 0.759 & 0.0341 & 0.0060 & 0.340 \\
Progressive pro-consumption & 0.284 & 2.025 & 0.899 & 0.0482 & 0.0085 & 0.735 \\
\bottomrule
\end{tabular}
\end{table}

Figure \ref{fig:rebate_income} shows the same positive-tax experiment by group. It is included because the rebate is numerically small in aggregate plots. The figure makes clear that the rebate is a separate household flow, but also that it is small relative to labor and capital income.

\begin{figure}[!htbp]
\centering
\includegraphics[width=0.88\textwidth]{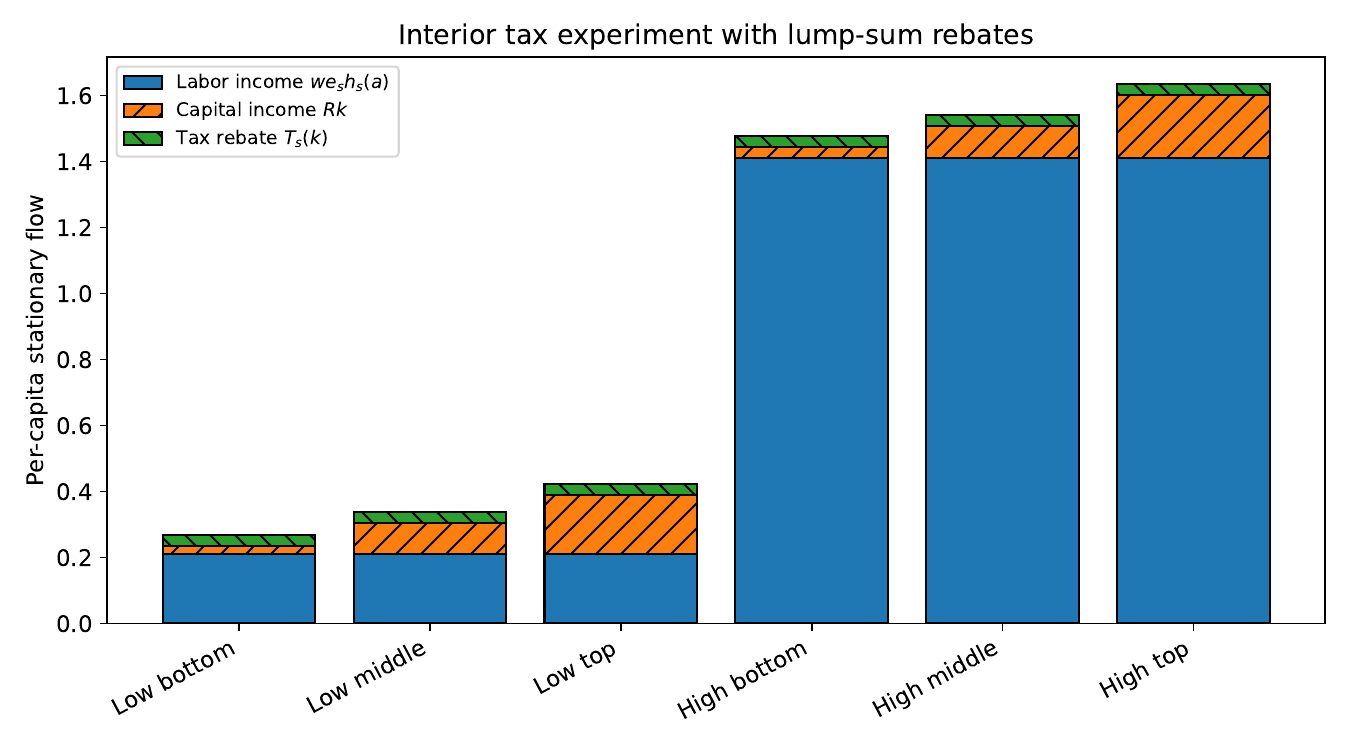}
\caption{Interior tax experiment with lump-sum rebates. The figure reports per-capita labor income \(we_sh_s(a)\), capital income \(Rk\), and the tax rebate \(T_s(k)\) for \(\tau=0.10\) and \(\omega_T=0.15\). Automation dividends are omitted because their per-capita flow is below \(2\times10^{-4}\) in this experiment and would not be visually informative.}
\label{fig:rebate_income}
\end{figure}

\begin{figure}[!htbp]
\centering
\includegraphics[width=0.88\textwidth]{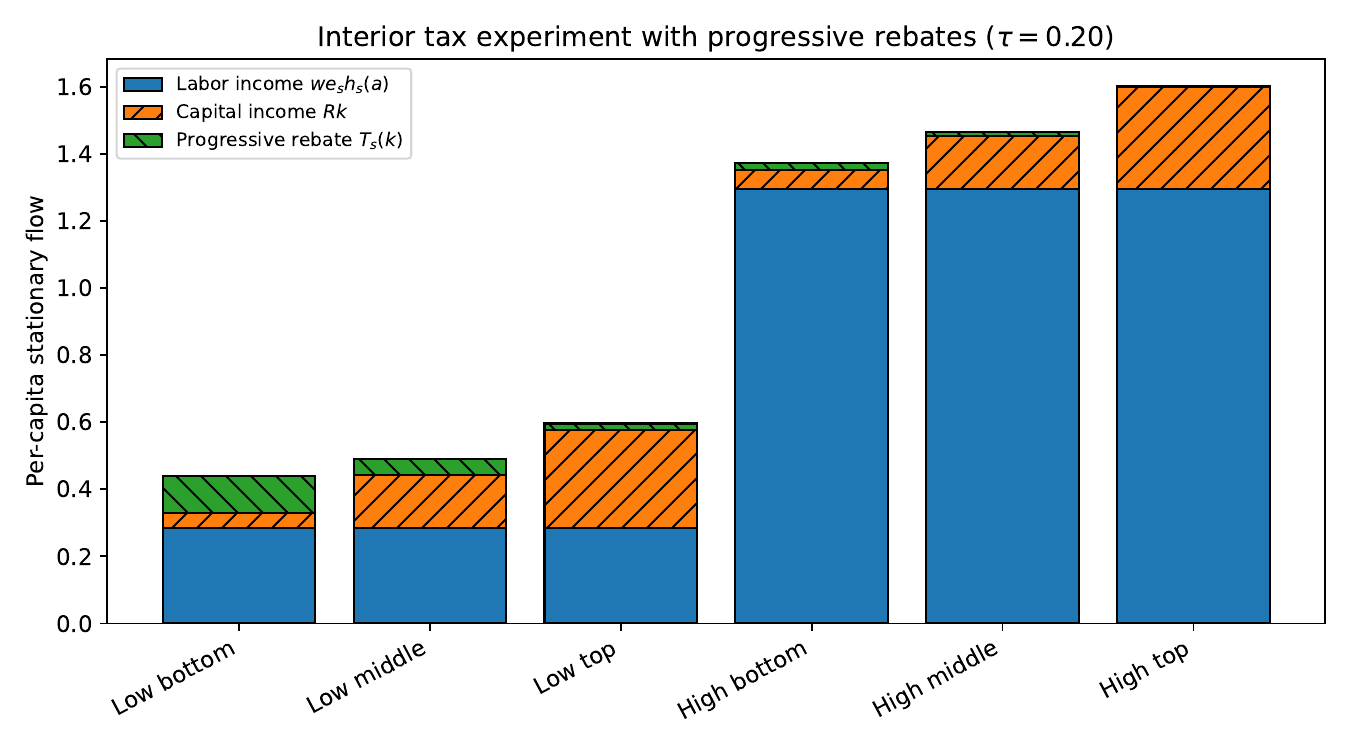}
\caption{Interior tax experiment with progressive rebates. The figure reports the same income components as Figure \ref{fig:rebate_income}, but the rebated revenue is targeted toward low-wealth and low-labor-income households. The tax is \(\tau=0.20\), the fiscal friction is \(\omega_T=0.15\), and the rebate kernel is proportional to \(\exp(-\varrho_k k-\varrho_y y_s^L)\) with \(\varrho_k=0.55\) and \(\varrho_y=2\).}
\label{fig:progressive_rebate_income}
\end{figure}

Figure \ref{fig:progressive_rebate_income} reports the corresponding progressive-rebate income components. The progressive rebate is visible in the low-skill bottom-wealth cell because the kernel targets low wealth and low labor income.

\paragraph{A pro-consumption progressive rebate.}
The government need not rebate revenue equally or in proportion to income. To illustrate a policy aimed directly at households with low savings and low earnings, define a progressive rebate kernel
\[
b_s^{P}(k)=
\frac{\exp[-\varrho_k k-\varrho_y y_s^L(a)]}
{\sum_{s'}\int \exp[-\varrho_k k'-\varrho_y y_{s'}^L(a)]g_{s'}(k')\,dk'}.
\]
This kernel gives more of the rebated revenue to low-wealth and low-labor-income households. For a skill-wealth cell \((s,b)\), where \(s\in\{U,H\}\) indexes skill and \(b\) indexes the bottom, middle, or top wealth bin, define the average transfer
\[
\overline T_{sb}=\frac{\int_{\mathcal B_b}T_s(k)g_s(k)\,dk}{\int_{\mathcal B_b}g_s(k)\,dk}.
\]
Table \ref{tab:progressive_rebate} reports this \(\overline T_{sb}\), the cell average wealth MPC \(\overline{\mathrm{MPC}}_{sb}\), and the first-round product \(\overline{\mathrm{MPC}}_{sb}\overline T_{sb}\). The example below uses \(\tau=0.20\), \(\omega_T=0.15\), \(\varrho_k=0.55\), and \(\varrho_y=2\). The equilibrium automation response is pinned down by the marginal tax, while the rebate rule reallocates the after-tax revenue. Table \ref{tab:progressive_rebate} and Figure \ref{fig:progressive_rebate} show that a targeted rebate delivers much larger direct consumption support to the low-skill bottom-wealth group because that group receives a larger transfer and has the highest local MPC.

\begin{table}[!htbp]
\centering
\caption{Progressive rebate case study, \(\tau=0.20\), \(\omega_T=0.15\)}
\label{tab:progressive_rebate}
\small
\begin{tabular}{lcccc}
\toprule
Rule and group & Transfer & MPC & MPC\(\times\)transfer & Interpretation \\
\midrule
Lump-sum, low bottom & 0.048 & 0.283 & 0.0137 & equal rebate \\
Progressive, low bottom & 0.112 & 0.283 & 0.0316 & targeted support \\
Lump-sum, high top & 0.048 & 0.142 & 0.0068 & equal rebate \\
Progressive, high top & 0.004 & 0.142 & 0.0006 & little support \\
\bottomrule
\end{tabular}
\end{table}

\begin{figure}[!htbp]
\centering
\includegraphics[width=0.95\textwidth]{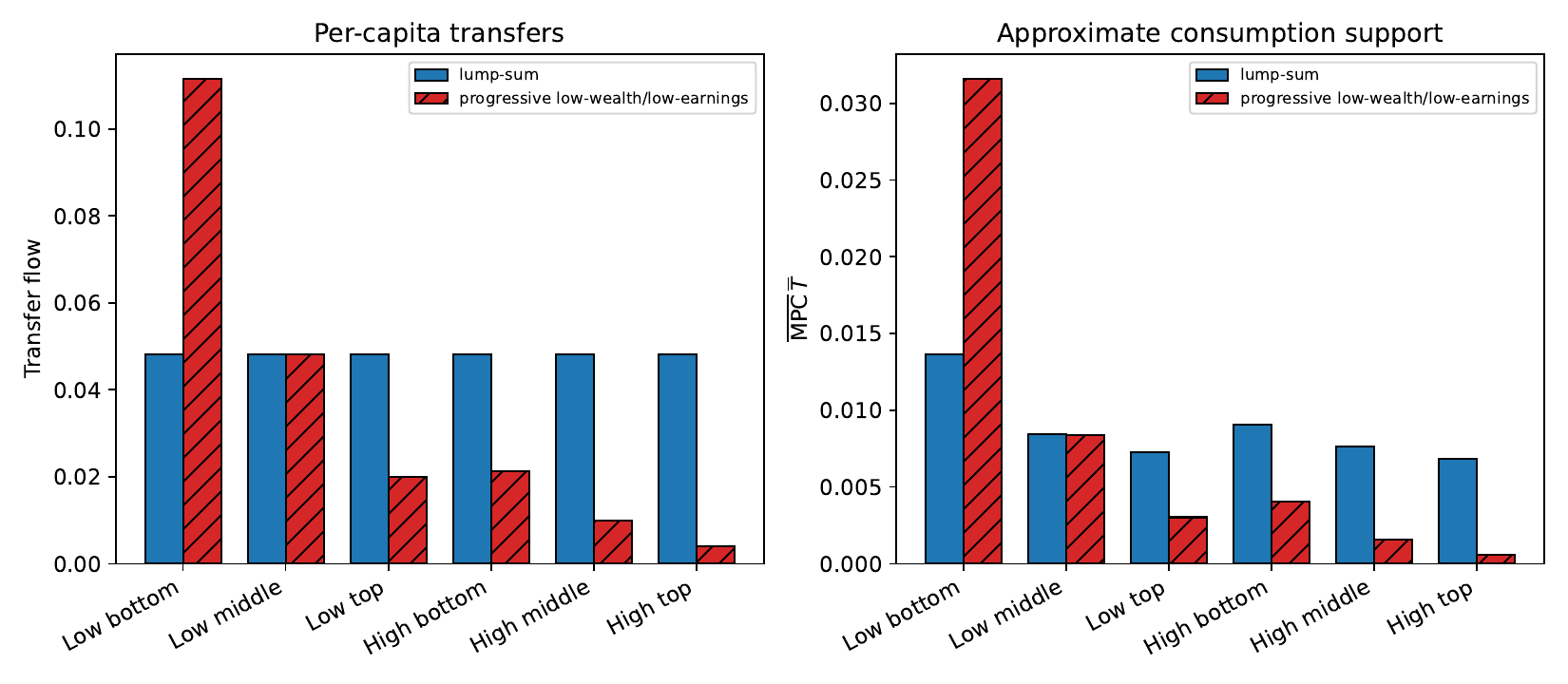}
\caption{Progressive rebate case study. The left panel reports per-capita transfers. The right panel multiplies the transfer by the local MPC to show approximate consumption support. The progressive rule supports low-wealth, low-earnings households much more strongly than a lump-sum rebate.}
\label{fig:progressive_rebate}
\end{figure}

Figure \ref{fig:progressive_targeted_diagnostics} reports a first-round consumption-support diagnostic for each skill-wealth cell \((s,b)\), where \(s\) is skill and \(b\) is the wealth bin:
\[
\text{support}_{sb}=\frac{\overline{\mathrm{MPC}}_{sb}\,\overline T_{sb}}{\bar c_{sb}},
\]
where \(\overline T_{sb}\) is the average tax rebate received by the cell, \(\overline{\mathrm{MPC}}_{sb}\) is the average local marginal propensity to consume computed from the HJB consumption policy, and \(\bar c_{sb}\) is average consumption in the same cell. The numerator approximates the immediate consumption response to the transfer, and the denominator scales that response by the cell's average consumption. The left bracket groups the low-skill cells and the right bracket groups the high-skill cells. Long-run effects continue to operate through saving, prices, and the stationary distribution.

\begin{figure}[!htbp]
\centering
\includegraphics[width=0.92\textwidth]{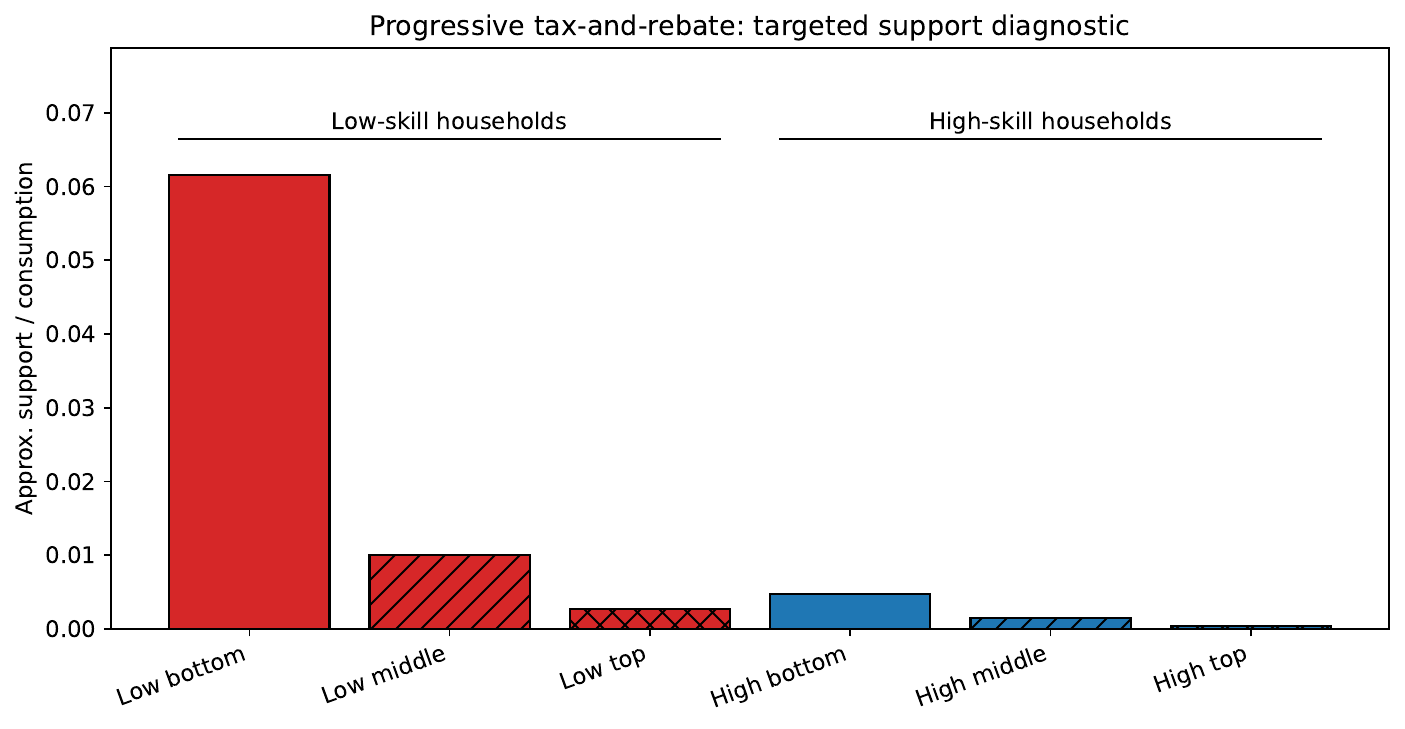}
\caption{Targeted-support diagnostics for the progressive rebate. Bars report \(\overline{\mathrm{MPC}}_{sb}\overline T_{sb}/\bar c_{sb}\). The left bracket groups low-skill households and the right bracket groups high-skill households.}
\label{fig:progressive_targeted_diagnostics}
\end{figure}

\subsection{Domestic ownership and the asset-return channel}

The model does not contain a separate stochastic stock-price index. The household asset is a claim to productive capital and domestic automation rents, with installed capital priced at one. Thus the asset-market objects are the aggregate capital stock \(K\), the productive-capital return \(r\), and the household asset return
\[
R=r+\theta_E\frac{\Pi^A}{K}.
\]
The dividend yield from domestic automation rents is
\[
\text{dividend yield}=R-r=\theta_E\frac{\Pi^A}{K}.
\]
I report both \(R\) and the dividend yield because \(R\) is the total return received by households, while \(R-r\) isolates the ownership channel. This is the channel through which wealthy households, including wealthy low-skilled households, may benefit from automation even when exposed labor income falls.

Domestic ownership matters because wealthy households own capital and automation rents. The main benchmark uses \(\theta_E=0.45\), while the table below reports the low-pass-through case \(0.15\) and higher pass-through cases. Table \ref{tab:ownership} varies \(\theta_E\), the domestic ownership share of net automation rents. Higher domestic ownership raises the household asset return \(R\) relative to the productive capital return \(r\), slightly increases capital accumulation in this calibration, and partially insulates owners. It does not eliminate adverse incidence for low-wealth exposed households, because they own little capital.

\begin{table}[!htbp]
\centering
\caption{Domestic ownership sensitivity, decentralized no-tax equilibrium}
\label{tab:ownership}
\begin{tabular}{ccccccc}
\toprule
$\theta_E$ & $a^{D}$ & $K^{D}$ & $r^{D}$ & $R^{D}$ & $C^{D}$ & div. yield \\
\midrule
0.00 & 0.523 & 1.993 & 0.0045 & 0.0045 & 0.602 & 0.0000 \\
0.15 & 0.524 & 2.007 & 0.0033 & 0.0050 & 0.604 & 0.0017 \\
0.30 & 0.525 & 2.022 & 0.0022 & 0.0055 & 0.607 & 0.0034 \\
0.45 & 0.526 & 2.036 & 0.0010 & 0.0060 & 0.609 & 0.0050 \\
0.60 & 0.527 & 2.050 & -0.0002 & 0.0065 & 0.611 & 0.0067 \\
\bottomrule
\end{tabular}
\end{table}

The aggregate effects of \(\theta_E\) in Table \ref{tab:ownership} are intentionally small in this calibration: changing ownership reallocates automation rents, but it does not directly change the firm's technology or the tax. For that reason I do not plot the aggregate series separately. The economically relevant object is the ownership wedge in household returns.

Figure \ref{fig:ownership_wedge} isolates the domestic-ownership channel. It does not say that ownership changes the technology much; instead it shows how ownership changes the return received by households. This is the channel through which high-wealth households can be protected even when exposed labor income falls.

\begin{figure}[!htbp]
\centering
\includegraphics[width=0.82\textwidth]{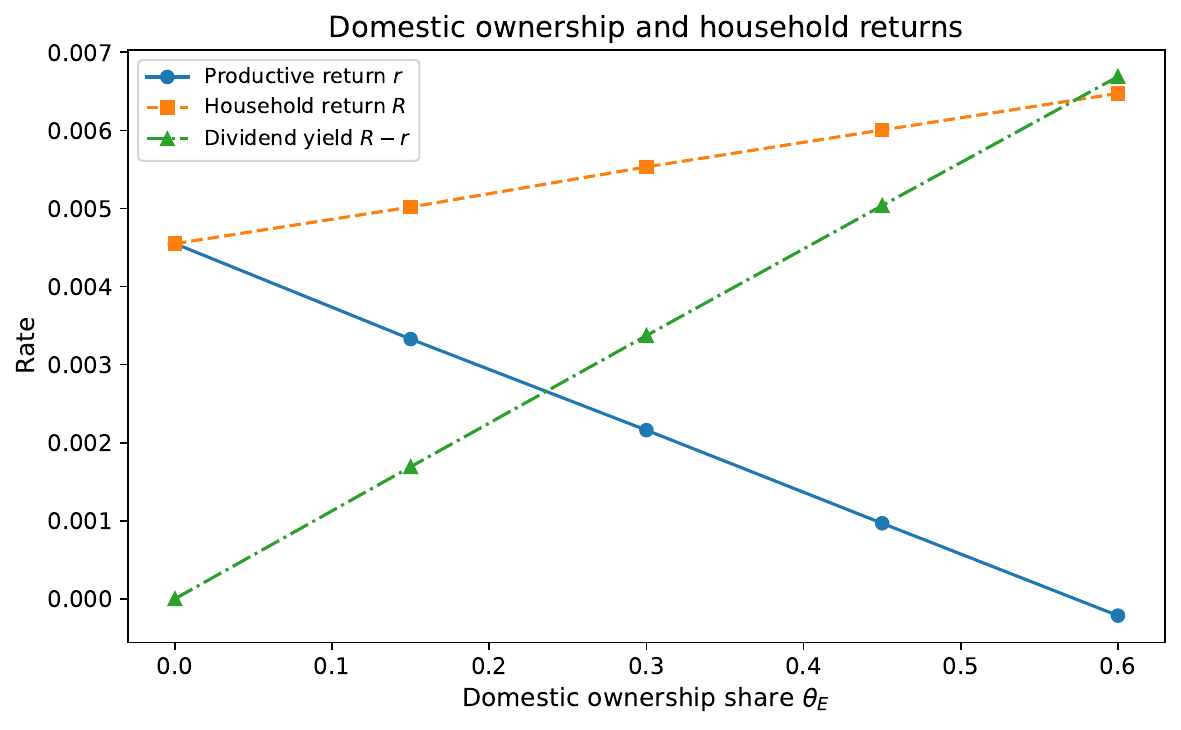}
\caption{Domestic ownership and household returns. The productive capital return \(r\) is the investment opportunity from the competitive firm. The household asset return is \(R=r+\theta_E\Pi^A/K\). The plotted wedge \(R-r\) is the channel through which domestic ownership insulates wealthy households.}
\label{fig:ownership_wedge}
\end{figure}

\subsection{Instrument comparison at the policy-index target}

The derivative of the firm's policy payment, \(P_a(a;X)\), enters its automation condition. If one takes the reduced-form index target \(a^P\) as a political-economy objective, an interior target is implemented only if
\[
P_a(a^P;X^P)=M(a^P;X^P)-\phi-\kappa a^P,
\]
where \(M(a;X)\) is the private marginal benefit from automation before policy. A literal automation tax has
\[
P(a;X)=\tau a,
\qquad P_a(a;X)=\tau.
\]
A retained-labor subsidy \(\sigma\) per unit of paid human task input has
\[
P(a;X)=\tau a-\sigma H(a;g),
\]
so, using the exposure index \(\Lambda_H\) defined in \eqref{eq:lambda_h},
\[
P_a(a;X)=\tau+\sigma\Lambda_H(a;g).
\]
Two policy schedules \(P\) and \(\widetilde P\) implement the same interior index target only if they generate the same marginal payment at the target state,
\[
P_a(a^P;X^P)=\widetilde P_a(a^P;X^P)=\iota^P.
\]
For example, a pure automation tax implements the target with \(\tau=\iota^P\). A pure retained-labor subsidy, with \(P^S(a;X)=-\sigma H(a;g)\), implements the same target if
\[
\sigma=\frac{\iota^P}{\Lambda_H(a^P;g^P)},
\qquad \Lambda_H(a^P;g^P)>0.
\]
The two instruments are equivalent only at the target margin; they need not have the same revenue, distributional incidence, or off-equilibrium effects. The calculation implements the reduced-form index target and allows the incidence of different instruments to be compared.

Fiscal closure determines who ultimately bears or receives the tax revenue. A tax can be rebated to households, dissipated through administrative frictions, spent on public goods, or transferred to another claimant. In the boundary implementation, revenue is zero because \(a=0\). In the interior experiments, revenue is positive and the rebate rule materially affects stationary incidence and household value comparisons.


\section{A Proxy Regime Diagnostic for the Current AI Economy}
\label{sec:empirical_regime_classification}

The model implies that the relevant empirical question is not only whether automation raises measured productivity, but whether the decentralized automation choice leaves the demand-weighted objective increasing or decreasing. Let \(a^{D}\) denote the private automation choice. Define
\[
D_{\mu}(a)
\equiv
\frac{\partial}{\partial a}\left[\lambda C(a)+\mu B_U(a)\right],
\]
where \(C(a)\) is aggregate household consumption after prices, the stationary distribution, and the interest rate have adjusted, and \(B_U(a)\) is the exposed wage bill defined in Section \ref{subsec:stationary_marginal_derivative}. In the empirical proxy below, changes in exposed labor income are the data-side analogue of changes in \(B_U(a)\). Under a local single-crossing condition, \(D_{\mu}(a^{D})>0\) places the policy-index target locally above the decentralized choice, while \(D_{\mu}(a^{D})<0\) places it locally below.

For the empirical implementation, let bins \(b\in\mathcal B\) index households by occupation, education, wealth, age, and industry. Let \(\eta_b\) denote the marginal propensity to consume of bin \(b\), let \(\omega^K_b\) denote the share of automation rents or equity gains accruing to that bin, and let \(\Delta_a y^L_{b,t}\), \(\Delta_a B_{U,t}\), and \(\Delta_a\Pi_t\) denote local changes in bin labor income, the exposed wage-bill analogue, and automation rents around the current automation margin. The empirical analogue is
\[
\widehat D_{\mu,t}
=
\lambda
\sum_{b\in\mathcal B}
\widehat\eta_b
\left[
\widehat{\Delta_a y^L_{b,t}}
+
\widehat\omega^K_{b,t}\widehat{\Delta_a\Pi_t}
\right]
+
\mu \widehat{\Delta_a B_{U,t}} .
\]

Table \ref{tab:empirical_regime_proxies} reports the proxy calibration for the contemporary United States. Current AI diffusion is no longer negligible, but it remains concentrated. The Census BTOS AI supplement reports AI use by 18\% of firms, or 32\% when employment-weighted, during November 2025--January 2026 \citep{BonneyEtAl2026AIDiffusion}. The Real-Time Population Survey reports work-related GenAI use by about 41\% of workers in November 2025, and the Survey of Business Uncertainty reports that 78\% of the labor force works at firms using AI \citep{Allen2026AIAdoption}. The private investment signal is strong, but broad pass-through remains mixed: BLS reports nonfarm business productivity growth of 0.8\% annualized in 2026Q1 and 2.9\% over the year, while real hourly compensation fell 0.5\% in the quarter and the labor share was 54.1\% \citep{BLSProductivity2026Q1}.

\begin{table}[!htbp]
\centering
\caption{Proxy calibration for the current AI regime}
\label{tab:empirical_regime_proxies}
\small
\begin{tabular}{p{0.25\textwidth}p{0.36\textwidth}p{0.27\textwidth}}
\toprule
Model object & Empirical proxy & Current calibration signal \\
\midrule
Automation intensity \(a_t\) & Firm and worker use of AI & AI use is incomplete but material: 18\% firm adoption, 32\% employment-weighted firm adoption, and 41\% employment-weighted worker-task use in the Census AI supplement \citep{BonneyEtAl2026AIDiffusion}. \\
\addlinespace
Private marginal benefit \(\mathcal M(a)\) & Equipment, software, and data-center investment & The private return signal is strong: 2026Q1 equipment investment rose 17.2\%, intellectual-property-products investment rose 13.0\%, and data-center investment rose more than 22\% annualized \citep{USTreasury2026TBAC}. \\
\addlinespace
Broad productivity pass-through & Labor productivity, compensation, and labor share & Productivity growth is positive but broad labor pass-through is incomplete: productivity rose 0.8\% annualized, real hourly compensation fell 0.5\%, and the labor share was 54.1\% \citep{BLSProductivity2026Q1}. \\
\addlinespace
Ownership of automation rents \(\omega^K_b\) & Corporate equity and mutual-fund ownership by wealth group & AI-rent claims are likely to be concentrated: in 2025Q3, the top 1\% held 50.2\% of corporate equities and mutual fund shares, the 90th--99th percentiles held 37.2\%, and the bottom 50\% held 1.1\% \citep{FREDTop1Equity2026,FRED9099Equity2026,FREDBottom50Equity2026}. \\
\addlinespace
MPC weights \(\eta_b\) & MPC heterogeneity by wealth and income & The Fed estimates an MPC out of wealth of 0.8 cents per dollar for the top income quintile and 7.5 cents per dollar for the other quintiles \citep{Beach2025WealthMPC}. \\
\bottomrule
\end{tabular}
\end{table}

Combining the ownership and MPC statistics gives the equity-rent demand pass-through
\[
\widehat\eta^K
=
0.874(0.008)+0.126(0.075)
=
0.016442 .
\]
Thus, one additional dollar of capitalized AI rent generates about 1.6 cents of consumption demand under this two-group mapping. Because direct microdata linking current AI use to identified labor-income changes are not yet sufficient to estimate \(\Delta_a y^L_b\) structurally, I classify the current economy with a normalized sensitivity calibration. Normalizing the AI capital-rent gain to one dollar, the baseline uses \(\widehat{\Delta_a Y^L_{\mathrm{broad}}}=0.005\), \(\widehat{\Delta_a B_U}=-0.025\), \(\lambda=1\), and \(\mu=1\). These values are transparent sensitivity inputs, not estimates of realized aggregate effects on GDP, consumption, employment, or total labor income. This gives
\[
\widehat D_{\mu,2026}^{\mathrm{base}}
= [0.016442+0.005]-0.025
= -0.003558 .
\]
The current economy is therefore close to the boundary but slightly on the adverse-incidence side. The classification is moderate-confidence: current evidence points to substantial task augmentation, while capitalized AI rents remain concentrated and broad labor-income pass-through has not yet clearly dominated the incidence channel.

Table \ref{tab:empirical_current_severity} gives the scale of the number. In the normalization used here, the positive demand channels add \(0.021442\): \(0.016442\) from equity-rent demand pass-through and \(0.005\) from broad labor-income pass-through. The exposed-worker loss subtracts \(0.025\). Thus the positive channels offset 85.8\% of the exposed loss, leaving a residual negative gap of \(0.003558\), or 14.2\% of the exposed loss. Because the calculation is normalized to one dollar of AI rents, the residual gap is 0.36 cents per dollar of normalized AI rents. If AI rents were 1\%, 5\%, or 10\% of GDP, the corresponding GDP-equivalent gaps would be 0.0036\%, 0.0178\%, and 0.0356\% of GDP. These are scale translations of the local diagnostic, rather than estimates of realized aggregate effects.

In economic terms, gains accruing to firms, investors, and complementary workers almost offset the calibrated loss of exposed workers, but do not fully do so. The economy is close to the boundary between the productivity-led and adverse-incidence scenarios.

\begin{table}[!htbp]
\centering
\caption{Scale of the current AI-regime statistic}
\label{tab:empirical_current_severity}
\small
\begin{tabular}{l r}
\toprule
Quantity & Value \\
\midrule
Positive demand channels, $\widehat\eta^K+\widehat{\Delta_a Y^L_{\mathrm{broad}}}$ & 0.0214 \\
Exposed labor-income loss, $|\widehat{\Delta_a Y^L_U}|$ & 0.0250 \\
Residual negative gap, $|\widehat D_\mu|$ & 0.0036 \\
Residual gap per dollar of normalized AI rent & 0.36 cents \\
Positive-channel offset rate & 85.8\% \\
Residual gap as share of exposed loss & 14.2\% \\
GDP-equivalent gap if AI rents equal 1\% of GDP & 0.0036\% of GDP \\
GDP-equivalent gap if AI rents equal 5\% of GDP & 0.0178\% of GDP \\
GDP-equivalent gap if AI rents equal 10\% of GDP & 0.0356\% of GDP \\
\bottomrule
\end{tabular}
\end{table}

Because \(\widehat D_{\mu,2026}^{\mathrm{base}}\) is a calibrated sufficient-statistic value rather than a regression coefficient, Table \ref{tab:empirical_uncertainty_bands} reports calibration uncertainty bands, not frequentist confidence intervals. Holding the exposed-group loss fixed at \(-0.025\), the statistic crosses zero if broad labor-income pass-through rises from the baseline \(0.005\) to \(0.008558\). Holding broad pass-through fixed at \(0.005\), the statistic crosses zero if the exposed wage-bill loss is milder than \(-0.021442\). Thus the baseline is only \(0.003558\) normalized units from the boundary, and the sign is informative but not yet statistically settled.

\begin{table}[!htbp]
\centering
\caption{Calibration uncertainty bands for the current AI regime}
\label{tab:empirical_uncertainty_bands}
\scriptsize
\setlength{\tabcolsep}{3pt}
\begin{tabular}{p{0.25\textwidth}ccc p{0.16\textwidth}}
\toprule
Assumption set & Broad labor term & Exposed labor term & $\widehat D_\mu$ & Sign result \\
\midrule
Baseline point & $0.005$ & $-0.025$ & $-0.0036$ & negative \\
Exposed-loss band & $0.005$ & $[-0.035, -0.015]$ & $[-0.0136, 0.0064]$ & crosses zero \\
Broad-pass-through band & $[0.000, 0.015]$ & $-0.025$ & $[-0.0086, 0.0064]$ & crosses zero \\
Joint labor-pass-through band & $[0.000, 0.015]$ & $[-0.035, -0.015]$ & $[-0.0186, 0.0164]$ & crosses zero \\
\bottomrule
\end{tabular}
\end{table}

\subsection{Interpretation of the current diagnostic}

The baseline value \(\widehat D_{\mu,2026}^{\mathrm{base}}=-0.003558\) is close to zero. It therefore motivates a boundary-regime experiment rather than a large intervention: monitor labor-income pass-through and rent concentration, and vary the strength of rent recycling with the sign and magnitude of the updated statistic. The next subsection specifies that experiment.


\subsection{Policy experiment: state-contingent AI-rent recycling}
\label{subsec:policy_ai_rent_recycling}

The proxy diagnostic suggests a state-contingent instrument experiment. In a productivity-led regime the surcharge is zero. In an adverse-incidence regime a surcharge on concentrated automation rents is paired with transfers to high-MPC households. The experiment changes the firm's first-order condition and the household incidence of automation rents at the same time.

A practical policy is an \emph{AI-rent rebate}: a surcharge on concentrated AI rents, paired with a refundable worker rebate or payroll-tax credit targeted to high-MPC workers and households in high-exposure labor-market cells. In the notation of the model, the surcharge is the tax wedge \(\tau\) in the firm's automation condition,
\[
\mathcal M(a,\tau)-\phi-\kappa a-\tau=0,
\]
and the rebate is a transfer rule that raises \(C(a)\) by directing revenue to bins with high \(\eta_b\). The policy does not tax ordinary capital accumulation or block productivity-improving AI. It taxes the part of the return that is most likely to be capitalized into concentrated equity claims and recycles the revenue toward the households whose spending has the largest aggregate-demand effect.

The implementable version has four components.
\begin{enumerate}[label=(\roman*)]
\item \textbf{Tax base.} Apply the surcharge only to large firms or consolidated groups with substantial AI-related intangible, software, data-center, or automation-capital income. The base should be incremental rents rather than gross AI investment: profits or markups above a pre-adoption benchmark, above-normal returns on AI-intensive intangible capital, or a formula apportionment based on AI-related capital expenditure and excess operating margin. This avoids penalizing ordinary productivity investment and avoids requiring the government to classify every task as automated.

\item \textbf{Revenue recycling.} Return the revenue through a refundable worker credit, an expanded earned-income-tax-credit-style payment, or a payroll-tax rebate. The target is not a universal transfer to all households, but a transfer weighted toward low-wealth, high-MPC workers and households in occupations or industries with high AI exposure. In the model this raises \(\sum_b\eta_bT_b\), where \(T_b\) is the transfer received by bin \(b\). The United States already uses refundable credits, including the earned income tax credit, to deliver income support to low- and moderate-income workers \citep{IRSEITC2026}.

\item \textbf{Exposure adjustment.} Let the transfer weight depend on observable exposure indicators: occupation-level AI exposure, industry AI adoption, local labor-market displacement, and earnings losses relative to a pre-AI trend. This makes the policy closer to wage insurance than to a static redistribution program. Workers who experience wage losses from AI adoption receive larger rebates, while workers in complementary occupations receive smaller or no additional payments.

\item \textbf{Automatic trigger.} Activate or strengthen the surcharge-rebate rule when \(\widehat D_{\mu,t}\) becomes negative, or when observable sufficient statistics move in the same direction: rapid AI investment, concentrated equity gains, weak labor-income pass-through, and declining labor share. If \(\widehat D_{\mu,t}>0\), the surcharge can be reduced or suspended. Thus the policy is state-contingent rather than permanently anti-automation.
\end{enumerate}

The rent-recycling experiment differs from a tax on AI inputs. An input tax can reduce productivity-enhancing adoption even when \(D_{\mu}(a^{D})>0\). The rent surcharge instead targets the distribution of income claims and is activated only when the proxy statistic is negative. In that case, a positive \(\tau\) moves the private automation choice toward \(a^P\), while the high-MPC rebate raises the consumption pass-through from the remaining automation surplus. The calculation is a model-based instrument comparison; implementation would require a measurable rent base and an empirically validated exposure rule.

\section{Replication code}

The replication code is available at
\[
\texttt{https://github.com/ebayr/automation-demand-externality}.
\]
The repository contains the Python scripts used to solve the stationary HJB--KFE systems, compute the tax and ownership experiments, generate the tables, and reproduce the figures. The skill-mobility robustness table is generated by resolving the same public solver after setting \(\zeta\in\{-0.75,0,0.75\}\) and clearing its solution cache between cases.

\section{Conclusion}

Automation changes stationary allocations through productivity, task substitution and complementarity, skill transitions, capital obsolescence, precautionary saving, and ownership of automation rents. In the adverse-incidence scenario, decentralized automation lowers low-skill income, consumption, and capital relative to the no-automation allocation. In the productivity-led scenario, output, consumption, and capital rise.

The mobility robustness exercise shows why these results should be stated conditionally. Holding automation fixed at the baseline intensity, reversing the skill-transition tilt raises the high-skill population share, stationary consumption, output, and capital substantially. The adverse mobility specification and the obsolescence parameter are therefore important scenario assumptions rather than secondary details.

The proxy empirical diagnostic places the current U.S. economy close to the boundary between the productivity-led and adverse-incidence scenarios. Positive demand channels offset most of the normalized exposed-worker loss, while concentrated equity ownership limits the consumption pass-through from AI rents. The diagnostic is deliberately reported with sensitivity bands because modest changes in broad labor-income pass-through reverse its sign.

The main quantitative conclusion is that the long-run effect of automation cannot be inferred from its productivity effect alone. Labor-income exposure, reskilling, obsolescence, and rent ownership jointly determine consumption, saving, and capital accumulation. The reduced-form policy index and the rent-recycling experiment show how alternative fiscal closures distribute the remaining gains. A full welfare analysis of transition paths is a distinct next step.

\appendix
\section{Computational Appendix}

\subsection{Stationary system}

For given \((a,\tau)\), the numerical equilibrium solves for
\[
\{V_s(k),c_s(k),\dot k_s(k),g_s(k),T_s(k)\}_{s\in\{U,H\}},
\qquad
K,L,H,C,Y,r,w,R.
\]
The HJB, KFE, and factor-pricing equations are those stated in the text. Capital-market clearing is
\[
K^{hh}(r,a,\tau)=\sum_s\int kg_s^{r,a,\tau}(k)\dd k=K^{firm}(r,a),
\]
where firm capital demand is obtained from
\[
r+\delta(a)=\alpha Z(a)K^{\alpha-1}L^{1-\alpha},
\]
so
\[
K^{firm}(r,a)=L(a;g)\left(\frac{\alpha Z(a)}{r+\delta(a)}\right)^{1/(1-\alpha)}.
\]
The goods-market residual reported by the code is
\[
\mathcal E^{goods}=Y-[C+\delta(a)K+\Phi^A(a)+\omega_T\tau a+(1-\theta_E)\Pi^A].
\]

\subsection{Numerical implementation parameters and upwind scheme}
\label{app:numerical_parameters}

The numerical parameters are kept out of Table \ref{tab:all_parameters} because they are not economic primitives. Table \ref{tab:numerical_parameters} reports the baseline grid and solver settings used for the stationary computations.

\begin{table}[!htbp]
\centering
\caption{Numerical implementation parameters}
\label{tab:numerical_parameters}
\small
\begin{tabular}{lll}
\toprule
Object & Symbol or code name & Value \\
\midrule
Asset grid size & \(J\) & 31 \\
Asset lower bound & \(\underline k\) & 0 \\
Asset upper bound & \(\bar k\) & 18 \\
Grid spacing & \(\Delta k\) & \(18/(31-1)=0.60\) \\
Automation grid for reported searches & \(\mathcal A\) & 61 points on \([0,0.90]\) \\
HJB implicit pseudo-time step & \(\Delta\) & 1000 \\
HJB maximum iterations & -- & 20 \\
HJB tolerance & -- & \(10^{-5}\) \\
Minimum consumption floor & -- & \(10^{-10}\) \\
Interest-rate root tolerance & -- & approximately \(3\times10^{-5}\) \\
Automation-root tolerance & -- & approximately \(5\times10^{-4}\) \\
\bottomrule
\end{tabular}
\end{table}

For a candidate \((r,a,\tau)\), the HJB is discretized by an implicit upwind finite-difference scheme. Let
\[
I_{s,j}=Rk_j+we_sh_s(a)+T_s(k_j)
\]
denote non-consumption income at grid point \((k_j,s)\). Given a value-function iterate \(V_{s,j}\), define forward and backward derivatives
\[
D^+V_{s,j}=\frac{V_{s,j+1}-V_{s,j}}{\Delta k},
\qquad
D^-V_{s,j}=\frac{V_{s,j}-V_{s,j-1}}{\Delta k}.
\]
The candidate consumption policies are
\[
c^+_{s,j}=(D^+V_{s,j})^{-1/\gamma},
\qquad
c^-_{s,j}=(D^-V_{s,j})^{-1/\gamma},
\]
and the corresponding drifts are
\[
s^+_{s,j}=I_{s,j}-c^+_{s,j},
\qquad
s^-_{s,j}=I_{s,j}-c^-_{s,j}.
\]
The upwind derivative uses \(D^+V\) when the drift points to the right, \(D^-V\) when it points to the left, and the state-constraint derivative \(u'(I_{s,j})\) when the drift is locally zero. Boundary derivatives are chosen so that the savings drift does not leave the grid:
\[
\dot k_s(k_1)\ge0,
\qquad
\dot k_s(k_J)\le0.
\]

After the upwind direction is chosen, write the selected drift as \(s_{s,j}=I_{s,j}-c_{s,j}\), with
\[
s_{s,j}^+=\max\{s_{s,j},0\},\qquad s_{s,j}^-=\min\{s_{s,j},0\}.
\]
The drift contribution to the generator is
\[
(\mathcal A_k v)_{s,j}
=\frac{s_{s,j}^+}{\Delta k}(v_{s,j+1}-v_{s,j})
+\frac{-s_{s,j}^-}{\Delta k}(v_{s,j-1}-v_{s,j}),
\]
with the boundary terms suppressed when \(j=1\) or \(j=J\) and the no-outflow condition imposed. The skill-switching contribution is
\[
(\mathcal A_q v)_{s,j}=\sum_{s'\ne s}q_{ss'}(a)(v_{s',j}-v_{s,j}),
\]
so the full generator is \(\mathcal A=\mathcal A_k+\mathcal A_q\). The implicit HJB step solves
\[
\left[(\rho+1/\Delta)I-\mathcal A\right]V^{new}
=u(c)+\frac{1}{\Delta}V^{old},
\]
iterating until the sup-norm change in \(V\) is below the tolerance in Table \ref{tab:numerical_parameters}.

The resulting drift matrix and the skill-transition matrix \(Q(a)\) define the finite-state generator \(\mathcal A(r,a,\tau)\). The stationary distribution solves
\[
\mathcal A(r,a,\tau)^\top g=0,
\qquad
\mathbf 1^\top g=1,
\]
where one row of the linear system is replaced by the normalization condition.

\subsection{Calibration mechanics and scenario construction}

The calibration is not a structural estimation exercise. It is a disciplined quantitative illustration and a set of scenario exercises. Each parameter is tied either to a standard macro value, an empirical guidepost, or a scenario dimension. The mechanics are as follows. First, the skill-efficiency ratio \(e_H/e_U\) is chosen to match the high-/low-education wage ratio discussed in the introduction. Second, the exposure and complementarity parameters \(\chi_U,\beta_H,\xi_U,\eta_H^L\) are chosen so that automation sharply reduces low-skilled paid tasks while increasing high-skilled task value. This creates the two-sided AI debate in the model: firms and high-skilled workers have a genuine reason to adopt, while exposed low-wealth households lose income. Third, \(\psi_Z\) is mapped into a full-index productivity gain through \(e^{\psi_Z}-1\). Fourth, \(\delta_A\) governs whether automation makes legacy capital obsolete. Fifth, \(\theta_E\) is varied because domestic ownership of AI rents is a distributional object. Sixth, the tax and rebate exercises vary \(\tau\), \(\omega_T\), and the rebate kernel \(b_s(k)\) while holding technology primitives fixed.

For each proposed parameter vector, the code solves the full stationary fixed point. The calibration is rejected if capital-market clearing fails, if the goods-market residual is not near zero, or if the decentralized automation residual does not have a well-defined root. Prices are functions of the distribution through market clearing, and the stationary distribution is obtained from the HJB--KFE system. Gu et al. (2024) emphasize precisely this structure: heterogeneous-agent equilibria require agent optimization, KFE consistency, and market-clearing prices as functions of aggregate states and distributions \citep{gu2024}.

Table \ref{tab:calibration_urls} lists the public URLs for the empirical guideposts used in the calibration discussion.

\begin{table}[!htbp]
\centering
\caption{Empirical guideposts and report URLs used in the calibration discussion}
\label{tab:calibration_urls}
\scriptsize
\begin{tabular}{p{0.21\textwidth}p{0.30\textwidth}p{0.39\textwidth}}
\toprule
Source & Object used in the paper & URL \\
\midrule
BLS Education Pays & Education wage premia used to discipline \(e_H/e_U\). & \url{https://www.bls.gov/careeroutlook/2025/data-on-display/education-pays.htm} \\
BLS CPS Table 37b & Employment and earnings by education. & \url{https://www.bls.gov/cps/cpsaat37b.htm} \\
Pew Research Center & AI exposure and earnings by occupation. & \url{https://www.pewresearch.org/social-trends/2023/07/26/earnings-of-workers-with-more-or-less-exposure-to-ai/} \\
Federal Reserve DFA & Distribution of corporate equity and mutual-fund ownership. & \url{https://www.federalreserve.gov/releases/z1/dataviz/dfa/distribute/chart/} \\
Treasury TIC & Foreign holdings of U.S. securities. & \url{https://home.treasury.gov/data/treasury-international-capital-tic-system} \\
Stanford AI Index & AI business-adoption and investment evidence. & \url{https://hai.stanford.edu/ai-index/2025-ai-index-report} \\
Goldman Sachs Research & Productivity-upside guidepost for generative AI. & \url{https://www.goldmansachs.com/insights/articles/generative-ai-could-raise-global-gdp-by-7-percent} \\
McKinsey Global Institute & Generative-AI productivity and adoption scenarios. & \url{https://www.mckinsey.com/mgi/our-research/the-economic-potential-of-generative-ai-the-next-productivity-frontier} \\
Gu--Lauriere--Merkel--Payne & Continuous-time heterogeneous-agent master-equation methods. & \url{https://arxiv.org/abs/2406.13726} \\
\bottomrule
\end{tabular}
\end{table}

\begin{enumerate}[label=(\roman*)]
\item \textbf{Production parameters.} The capital share is \(\alpha=0.36\), a standard Cobb--Douglas value. Baseline depreciation is \(\delta_0=0.06\). Automation-induced obsolescence \(\delta_A\) is varied because AI may make legacy capital less useful in the adverse-incidence regime but need not do so in the productivity-led regime.
\item \textbf{Preferences.} The coefficient of relative risk aversion is \(\gamma=2\). The discount rate \(\rho=0.15\) is the household discount rate in the stationary finite-grid economy. This value is not meant to be an annual time-preference estimate. It is a numerical calibration for the model period used here: with the coarse asset grid, strong automation-induced labor-income risk, and a single capital/equity asset, much lower values of \(\rho\) push precautionary saving and aggregate capital toward the grid boundary. The parameter is therefore chosen so that the computed economy has a nondegenerate wealth distribution, an interior market-clearing interest rate, and a stable HJB--KFE solution. The admissible choice is not arbitrary: \(\rho>0\), the market-clearing interest rate must lie inside the searched interval, the no-outflow boundary conditions must hold, and the goods-market residual must be close to zero.
\item \textbf{Skill efficiencies.} The ratio \(e_H/e_U=1.67\) is chosen to match the approximate earnings ratio between bachelor's-degree workers and high-school workers in Bureau of Labor Statistics education-earnings data. It is also close to the Pew high-AI-exposure versus low-exposure wage comparison.
\item \textbf{Automation technology.} The productivity function is \(Z(a)=Z_0e^{\psi_Za}\). A full-index productivity gain \(g_Z\) maps into \(\psi_Z=\log(1+g_Z)\). The baseline uses a moderate gain, while the productivity-led case uses a larger \(\psi_Z\). The functions \(h_U,h_H,\ell_U,\ell_H\) encode low-skill exposure and high-skill complementarity; they are varied across the two AI regimes.
\item \textbf{Automation costs.} The parameters \(\phi\) and \(\kappa\) govern linear and convex real costs of automation. They are set so that the decentralized baseline has an automation index near one half. In the productivity-led counterfactual, \(\kappa\) is higher to avoid pushing the model mechanically to the upper automation boundary when productivity gains are large.
\item \textbf{Skill mobility.} The parameters \(q_0\) and \(\zeta\) govern automation-dependent skill transitions. The baseline \(\zeta>0\) is an adverse-reskilling scenario; Table \ref{tab:mobility_robustness} also reports neutral and reversed tilts.
\item \textbf{Policy index.} The parameters \(\lambda\) and \(\mu\) are not preference primitives of households. They define a reduced-form political-economy index over aggregate consumption \(C\) and the exposed collective wage bill \(B_U\); that index is not a welfare criterion.
\item \textbf{Ownership and fiscal closure.} The central pass-through value is \(\theta_E=0.45\). The robustness cases \(0.15\) and \(0.75\) represent low and high domestic pass-through of automation rents. The fiscal friction \(\omega_T=0.15\) is a scenario parameter measuring administrative or political-economy leakage from tax revenue.
\end{enumerate}

\subsection{Pseudo-code}

\begin{enumerate}[label={\arabic*.},leftmargin=0.9cm]
\item Choose an asset grid and an automation grid.
\item For each \((a,\tau)\), guess \(r\).
\item Compute \(K^{firm}(r,a)\), \(w(r,a)\), labor income \(we_sh_s(a)\), net automation rent \(\Pi^A\), dividend yield \(\theta_E\Pi^A/K\), tax revenue \(\tau a\), and rebate schedule \(T_s(k)\).
\item Solve the HJB by implicit upwind finite differences.
\item Construct the generator for \((k,s)\) and solve the stationary KFE.
\item Compute household capital supply \(K^{hh}\).
\item Update \(r\) until \(K^{hh}=K^{firm}\).
\item Store \(K,L,H,C,Y,B,B_U,r,w,R,g,V,c\) and the goods-market residual.
\item Solve the decentralized automation residual.
\item Compute the reduced-form policy-index target by maximizing \(G_\mu(a)\).
\item Compute incidence by skill and wealth bins using the decentralized distribution as weights.
\end{enumerate}

An aggregate-shock extension would add a law of motion for the distribution and a master-equation block to this stationary HJB--KFE computation.

\subsection{Diagnostic no-wealth-heterogeneity benchmark}
\label{app:diagnostic_benchmark}

This appendix gives a no-wealth-heterogeneity benchmark that serves as a diagnostic comparison for the quantitative model. The benchmark suppresses precautionary saving, wealth-dependent marginal propensities to consume, ownership of automation rents, fiscal rebates, and the stationary wealth distribution. Those objects are present in the finite-grid heterogeneous-agent economy of Appendix \ref{app:finite_grid_equilibrium}. The purpose of the simplified environment is transparency: capital, wages, and skill masses can be eliminated analytically, so the derivative of the reduced-form policy index can be read from primitive exposure and productivity parameters.

Skill transitions are taken as primitive,
\[
q_{UH}(a)=\bar q_{UH}e^{-\zeta_U a},
\qquad
q_{HU}(a)=\bar q_{HU}e^{\zeta_H a},
\]
with stationary masses
\[
m_U(a)=\frac{q_{HU}(a)}{q_{UH}(a)+q_{HU}(a)},
\qquad
m_H(a)=1-m_U(a).
\]
Production-task labor and paid-task labor are
\[
L(a)=m_U(a)e_Ue^{-\xi_Ua}+m_H(a)e_He^{\eta_H^La},
\]
\[
H(a)=m_U(a)e_Ue^{-\chi_Ua}+m_H(a)e_He^{\beta_Ha}.
\]
Technology and depreciation are
\[
Z(a)=Z_0e^{\psi_Za},
\qquad
\delta(a)=\delta_0+\delta_Aa.
\]
In this benchmark the representative Euler condition pins down the productive-capital return at a required return \(\bar r\). Cobb--Douglas factor pricing then gives
\[
K(a)=L(a)\left(\frac{\alpha Z(a)}{\bar r+\delta(a)}\right)^{1/(1-\alpha)},
\]
\[
w(a)=(1-\alpha)Z(a)^{1/(1-\alpha)}
\left(\frac{\alpha}{\bar r+\delta(a)}\right)^{\alpha/(1-\alpha)}.
\]
Thus wages and capital are equilibrium objects in the diagnostic benchmark; they are not held fixed. Group labor incomes are
\[
y_U^L(a)=w(a)e_Ue^{-\chi_Ua},
\qquad
y_H^L(a)=w(a)e_He^{\beta_Ha}.
\]
The local consumption approximation is
\[
C(a)=C_0+\eta_Um_U(a)y_U^L(a)+\eta_H^Cm_H(a)y_H^L(a),
\qquad \eta_U,\eta_H^C\in[0,1],
\]
and the reduced-form policy index is
\[
G_\mu(a)=\lambda C(a)+\mu m_U(a)y_U^L(a).
\]
Define
\[
\omega_w(a)=\frac{d\log w(a)}{da}
=\frac{\psi_Z}{1-\alpha}
-\frac{\alpha}{1-\alpha}\frac{\delta_A}{\bar r+\delta(a)},
\]
so that
\[
\frac{d\log y_U^L(a)}{da}=\omega_w(a)-\chi_U,
\qquad
\frac{d\log y_H^L(a)}{da}=\omega_w(a)+\beta_H.
\]
Finally set
\[
\Gamma_U(a)=(\lambda\eta_U+\mu)y_U^L(a)
\left[m_U'(a)+m_U(a)(\omega_w(a)-\chi_U)\right],
\]
\[
\Gamma_H(a)=\lambda\eta_H^Cy_H^L(a)
\left[m_H'(a)+m_H(a)(\omega_w(a)+\beta_H)\right].
\]
The private automation residual in the diagnostic benchmark is
\[
F(a)=\mathcal M(a)-\phi-\kappa a,
\qquad
\mathcal M(a)=
\left[\psi_Z+(1-\alpha)\frac{L'(a)}{L(a)}\right]Y(a)+w(a)[-H'(a)].
\]

\paragraph{Derivative decomposition.}
In the no-wealth-heterogeneity benchmark,
\[
G_\mu'(a)=\Gamma_U(a)+\Gamma_H(a).
\]
If the private residual \(F\) has a unique root \(a^{D}\) on an interval \(I\), the policy index \(G_\mu\) is single-peaked, and
\[
\Gamma_U(a)^- > \Gamma_H(a)^+
\qquad\text{for all }a\in I,
\]
then \(a^P<a^{D}\). Conversely, if
\[
\Gamma_H(a)^+ > \Gamma_U(a)^-
\qquad\text{for all }a\in I,
\]
and the positive productivity/high-skill term dominates on the relevant interval, then \(a^P>a^{D}\) whenever the private root lies in that interval.

The stationary masses follow from the invariant distribution of the two-state Markov chain. Since the skill masses are functions of \(a\), both the level and the derivative of \(L(a)\) and \(H(a)\) are determined by primitive transition and exposure parameters. The Euler condition \(r(a)=\bar r\), together with the Cobb--Douglas capital first-order condition
\[
\bar r+\delta(a)=\alpha Z(a)K(a)^{\alpha-1}L(a)^{1-\alpha},
\]
gives the displayed expression for \(K(a)\). Substituting this expression into the wage first-order condition
\[
w(a)=(1-\alpha)Z(a)K(a)^\alpha L(a)^{-\alpha}
\]
gives the displayed expression for \(w(a)\), so wages and capital adjust within the diagnostic benchmark.

Differentiating \(\log w(a)\) gives \(\omega_w(a)\). Differentiating group labor incomes gives the two labor-income elasticities. Differentiating
\[
G_\mu(a)=\lambda\left[C_0+\eta_Um_U(a)y_U^L(a)+\eta_H^Cm_H(a)y_H^L(a)\right]
+\mu m_U(a)y_U^L(a)
\]
yields \(G_\mu'(a)=\Gamma_U(a)+\Gamma_H(a)\). A negative derivative on the interval containing the private root places the index maximizer to its left; a positive derivative places the maximizer to its right.

\subsection{Finite-grid existence and conditional uniqueness}\label{app:finite_grid_equilibrium}

The result below is deliberately stated for the finite-grid stationary economy that is actually computed. It is not a global uniqueness theorem for the infinite-dimensional heterogeneous-agent economy.

For fixed \((r,a,\tau)\), let \(\mathcal A(r,a,\tau)\) denote the finite-state generator over grid points \((k_j,s)\), and let \(g^{r,a,\tau}\) denote its invariant distribution. Define
\[
\mathcal E_K(r,a,\tau)=K^{hh}(r,a,\tau)-K^{firm}(r,a).
\]
After market clearing in \(r\), define the reduced automation residual
\[
\mathcal E_a(a,\tau)=\mathcal M(a,\tau)-\phi-\kappa a-\tau,
\]
where \(\mathcal M(a,\tau)\) is the private marginal automation benefit after the stationary distribution and interest rate have adjusted.

\begin{proposition}[Finite-grid stationary equilibrium]
Fix a compact asset grid \(\mathcal K_J=\{k_1,\ldots,k_J\}\), an automation interval \([0,\bar a]\), and an interest-rate interval \(\mathcal R=[\underline r,\overline r]\subset(-\delta(\bar a),\rho)\). Suppose:
\begin{enumerate}[label=(\roman*)]
\item for each \((r,a,\tau)\), the finite-grid discounted HJB has a unique solution and the induced policy is continuous in \((r,a,\tau)\);
\item the upwind scheme satisfies no-outflow boundary conditions, \(\dot k_s(k_1)\ge0\) and \(\dot k_s(k_J)\le0\);
\item \(q_{UH}(a)>0\) and \(q_{HU}(a)>0\) for all \(a\in[0,\bar a]\);
\item for each \(a\), \(\mathcal E_K(\cdot,a,\tau)\) is continuous and changes sign on \(\mathcal R\);
\item the reduced automation residual \(\mathcal E_a(\cdot,\tau)\) is continuous and changes sign on \([0,\bar a]\).
\end{enumerate}
Then a finite-grid stationary equilibrium exists. If \(\mathcal E_K(\cdot,a,\tau)\) is strictly monotone for each \(a\), and \(\mathcal E_a(\cdot,\tau)\) is strictly decreasing on \([0,\bar a]\), then the market-clearing interest rate and decentralized automation rate are unique.
\end{proposition}

\begin{proof}
Fix \((r,a,\tau)\). The finite-grid HJB is a discounted finite-state control problem. The discount rate \(\rho>0\) and assumption (i) give a unique value function and policy. Assumption (ii) keeps the savings drift inside the compact grid. Assumption (iii) makes the skill process irreducible; combined with the upwind transition matrix, the generator has an invariant distribution on its communicating class. Hence \(K^{hh}(r,a,\tau)\) and \(K^{firm}(r,a)\) are well defined.

By assumption (iv), \(\mathcal E_K\) is continuous in \(r\) and changes sign on a compact interval. The intermediate value theorem gives a market-clearing interest rate. Strict monotonicity gives uniqueness of that interest rate. Given the market-clearing interest rate for each \(a\), the reduced automation residual is well defined. Assumption (v) and the intermediate value theorem give an automation root. Strict monotonicity of \(\mathcal E_a\) gives uniqueness. Combining the HJB policy, invariant distribution, market-clearing interest rate, and automation root yields a finite-grid stationary equilibrium.
\end{proof}

For the reported calibration, the finite-grid conditions are checked directly. Table \ref{tab:existence_checks} reports the residual and derivative checks on the automation grid. The reduced residual changes sign between \(a=0.510\) and \(a=0.525\), and its numerical derivative is negative throughout the grid. The goods-market residuals are below \(10^{-7}\) in absolute value.

\begin{table}[!htbp]
\centering
\caption{Numerical verification of finite-grid equilibrium conditions}
\label{tab:existence_checks}
\begin{tabular}{lc}
\toprule
Object & Value \\
\midrule
\(\mathcal E_a(0,0)\) & 0.5887 \\
\(\mathcal E_a(0.90,0)\) & -0.2483 \\
Root bracket for \(a^{D}\) & [0.510, 0.525] \\
\(\max_a \partial \mathcal M(a,0)/\partial a\) on grid & -0.0502 \\
\(\max_a \partial \mathcal E_a(a,0)/\partial a\) on grid & -0.5702 \\
Goods residual, decentralized allocation & \(-5.85\times10^{-8}\) \\
Goods residual, policy-index allocation & \(-4.44\times10^{-10}\) \\
\bottomrule
\end{tabular}
\end{table}

\begin{figure}[!htbp]
\centering
\includegraphics[width=0.82\textwidth]{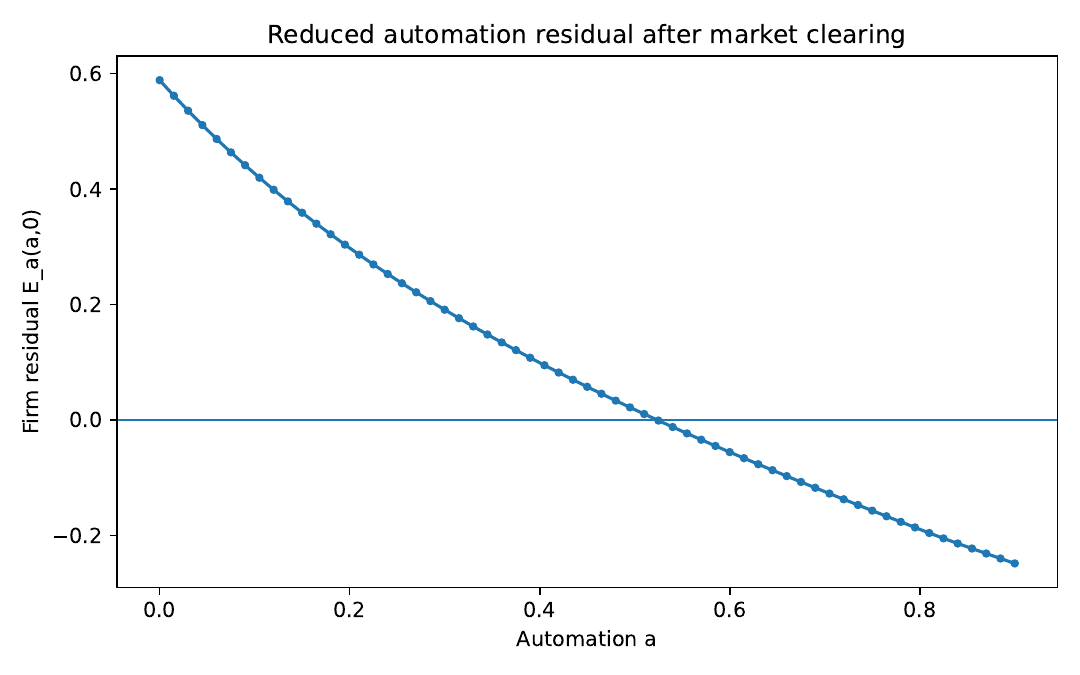}
\caption{Reduced automation residual after market clearing. The residual crosses zero once on the computed grid and is decreasing throughout the grid.}
\label{fig:residual_check}
\end{figure}

Figure \ref{fig:residual_check} visualizes the reduced automation residual used in Table \ref{tab:existence_checks}. A sufficient monotonicity condition is
\[
\kappa>\sup_{a\in[0,\bar a]}\frac{\partial \mathcal M(a,\tau)}{\partial a}.
\]
This inequality says that the convex marginal cost is large enough, relative to the slope of the equilibrium marginal benefit, to make the firm's reduced first-order-condition residual strictly decreasing. Hence the decentralized automation choice is unique whenever the residual crosses zero.

The monotonicity condition is used to obtain a unique interior decentralized automation level. If the reduced marginal benefit of automation is not decreasing sufficiently fast, the model may instead admit multiple stationary automation equilibria or a boundary solution. Economically, this corresponds to increasing returns, threshold effects, or organizational complementarities in AI adoption. In that case \(a^{D}\) should be defined as a global maximizer of firm value over the compact automation interval, and the numerical exercise should report all local roots or local maximizers when they exist. Multiple roots correspond to multiple AI-adoption regimes; an analysis of equilibrium selection across those regimes is left for future work.


\subsection{Empirical statistics and validation tests for the proxy regime diagnostic}
\label{app:regime_statistics_validation}

This appendix describes the statistics used to classify the current AI economy and the validation tests used to check the proxy regime diagnostic. The exercise is deliberately a sign test rather than a full structural estimate. The structural model determines which sign matters; the empirical exercise uses observable proxies to classify the sign of the local statistic.

\subsubsection{Statistic}

The target statistic is
\[
D_{\mu}(a^{D})
=
\left.
\frac{\partial}{\partial a}
\left[
\lambda C(a)+\mu B_U(a)
\right]
\right|_{a=a^{D}} .
\]
The economy is classified as a productivity-complementarity regime if \(D_{\mu}(a^{D})>0\), an adverse-incidence regime if \(D_{\mu}(a^{D})<0\), and a boundary regime if \(D_{\mu}(a^{D})\) is close to zero. In the empirical implementation, household cells \(b\in\mathcal B\) are indexed by occupation, education, wealth, age, and industry. The statistic is approximated by
\[
\widehat D_{\mu,t}
=
\lambda
\sum_{b\in\mathcal B}
\widehat\eta_b
\left[
\widehat{\Delta_a y^L_{b,t}}
+
\widehat\omega^K_{b,t}\widehat{\Delta_a\Pi_t}
\right]
+
\mu \widehat{\Delta_a B_{U,t}} .
\]
The first bracket captures the demand value of labor-income changes and capital-rent changes. The final term captures distributional concern for the exposed wage bill, the empirical analogue of \(B_U(a)\). This decomposition is useful because the available data are strongest for AI adoption, investment, equity ownership, and MPC heterogeneity, and weakest for the causal current-period mapping from AI adoption to labor-income changes.

\subsubsection{Proxy construction}

The statistic uses five groups of proxies.
\begin{enumerate}[label=(\roman*)]
\item \textbf{Automation intensity.} Firm-level and worker-level AI adoption measures come from the Census BTOS AI supplement and the Real-Time Population Survey. The Census supplement reports AI use by 18\% of firms during November 2025--January 2026, 32\% when employment-weighted, and worker-task AI use by firms representing 41\% of employment \citep{BonneyEtAl2026AIDiffusion}. The Real-Time Population Survey reports work-related GenAI use by about 41\% of workers in November 2025 \citep{Allen2026AIAdoption}.

\item \textbf{Private marginal benefit.} The calibration proxies \(\mathcal M(a)\) with AI-related investment. It uses 2026Q1 growth in real business equipment investment, intellectual-property-products investment, and data-center investment \citep{USTreasury2026TBAC}.

\item \textbf{Broad pass-through.} Labor productivity, real compensation, and labor share are aggregate pass-through diagnostics. In 2026Q1, nonfarm business labor productivity rose 0.8\% annualized and 2.9\% year over year, but real hourly compensation fell 0.5\% and the labor share was 54.1\% \citep{BLSProductivity2026Q1}.

\item \textbf{Ownership of automation rents.} The calibration uses Federal Reserve Distributional Financial Accounts shares of corporate equities and mutual fund shares. In 2025Q3, the top 1\% held 50.2\%, the 90th--99th wealth percentiles held 37.2\%, and the bottom 50\% held 1.1\% \citep{FREDTop1Equity2026,FRED9099Equity2026,FREDBottom50Equity2026}.

\item \textbf{MPC weights.} MPC-out-of-wealth estimates discipline the demand value of capitalized AI rents. The calibration uses Federal Reserve estimates of 0.8 cents per dollar for the top income quintile and 7.5 cents per dollar for the other quintiles \citep{Beach2025WealthMPC}.
\end{enumerate}

Combining the ownership and MPC statistics gives the equity-rent pass-through
\[
\widehat\eta^K
=
(0.502+0.372)(0.008)
+
\bigl[1-(0.502+0.372)\bigr](0.075)
=0.016442 .
\]
This means that a dollar of capitalized AI rent produces about 1.6 cents of consumption demand under the two-group mapping used in the calibration. The low value is a consequence of rent concentration among low-MPC equity owners.

\subsubsection{Baseline sign diagnostic}

The baseline sensitivity calculation normalizes the capital-rent gain to one dollar and sets \(\lambda=\mu=1\). Because the causal current-period labor-income effect of AI is not yet separately identified, the baseline is boundary-oriented: a small positive broad labor-income effect, \(0.005\), and a modest exposed wage-bill loss, \(-0.025\). These are normalized sensitivity inputs, not estimates of realized aggregate GDP, consumption, employment, or total labor-income effects. The resulting statistic is
\[
\widehat D_{\mu,2026}^{\mathrm{base}}
=
[0.016442+0.005]-0.025
=-0.003558 .
\]
The current economy is therefore classified as close to the boundary but slightly on the adverse-incidence side. Its sensitivity to the labor-income inputs motivates the calibration bands reported next.

\subsubsection{Calibration uncertainty bands}

The confidence language in the main text should be read as calibration uncertainty, not as a conventional confidence interval. The current exercise does not estimate \(\widehat D_{\mu,t}\) from a sampling design with standard errors. Instead, Table \ref{tab:empirical_uncertainty_bands} holds the equity-rent pass-through \(\widehat\eta^K=0.016442\) fixed and varies the two labor-pass-through inputs around the baseline. The table is generated from \texttt{data/current\_us\_uncertainty\_bands.csv}. It shows that the baseline negative value is close to zero: modestly stronger broad labor pass-through or modestly milder exposed-worker losses would move the statistic to the productivity-complementarity side. This is why the empirical classification is stated as near-boundary and slightly adverse-incidence rather than as a precise estimated sign.

\subsubsection{Sensitivity tests}

Table \ref{tab:appendix_regime_sensitivity} reports the sensitivity tests. The optimistic augmentation case assumes broad labor complementarity and small exposed-group losses, so the statistic is positive. The pessimistic displacement case assumes negative broad labor pass-through and larger exposed-group losses, so the statistic is negative. The broad productivity-dividend case mimics an electrification-style episode, while the neutral boundary case mechanically sets the statistic close to zero.

\begin{table}[!htbp]
\centering
\caption{Sensitivity tests for the empirical regime statistic}
\label{tab:appendix_regime_sensitivity}
\small
\begin{tabular}{l r l}
\toprule
Scenario & \(\widehat D_\mu\) & Classification \\
\midrule
Baseline near-boundary adverse incidence & -0.0036 & Adverse-incidence: \(a^P<a^{D}\) \\
Augmentation optimistic & 0.0439 & Productivity-complementarity: \(a^P>a^{D}\) \\
Displacement pessimistic & -0.0336 & Adverse-incidence: \(a^P<a^{D}\) \\
Broad productivity dividend & 0.0764 & Productivity-complementarity: \(a^P>a^{D}\) \\
Neutral boundary & 0.0000 & Boundary: \(D_\mu\approx0\) \\
\bottomrule
\end{tabular}
\end{table}

The sensitivity tests show that AI adoption alone does not determine the classification. The sign changes when broad labor-income pass-through becomes sufficiently positive. With concentrated rent ownership and incomplete labor-income pass-through, the current calibration lies slightly on the adverse-incidence side of the boundary.

\subsubsection{Historical validation tests}

The historical validation test asks whether the classifier assigns familiar automation episodes to plausible regimes. This appendix implements the exercise as a proxy backtest rather than only as a qualitative sign check. The backtest uses the same reduced statistic,
\[
\widehat D_\mu^{hist}
= \widehat{\Delta_a C}^{K}
+\widehat{\Delta_a Y}^{L,broad}
+\widehat{\Delta_a Y}^{L,exposed},
\]
where the first term is the demand value of capital-rent pass-through, the second term is broad labor-income pass-through, and the third term is the exposed-group labor-income term. This is still not a full structural historical estimation exercise: the historical episodes do not have the same household-bin microdata as the current U.S. AI episode, and the units of automation differ across robots, electrification, computerization, and steam mechanization. The object is a transparent proxy estimate used to test the classifier's sign.

The current-data calibration fixes the capital-rent demand pass-through at
\[
\widehat\eta^K
=(0.502+0.372)(0.008)+[1-(0.502+0.372)](0.075)
=0.016442 .
\]
For the current AI row, the backtest uses the same baseline inputs as the main empirical classification:
\[
0.016442+0.005-0.025=-0.003558 .
\]
For the industrial-robot row, the calculation uses the reduced-form estimates in \citet{acemoglu2020robots}: one additional robot per thousand workers reduces wages by 0.42 percent and the employment-population ratio by 0.2 percentage points. With a baseline employment-population ratio of 0.60, the employment margin is approximately \(-0.002/0.60=-0.0033\). Because the paper does not separately estimate a capital-rent pass-through term, the robot backtest uses the observable labor-income channel only:
\[
-0.0042-0.0033=-0.0075 .
\]
For older episodes without directly comparable microdata, the backtest uses conservative proxy values disciplined by the sign evidence: the industrial revolution is coded as near-boundary because productivity gains were real but early wage pass-through lagged \citep{Crafts2022IndustrialWages}; electrification is coded as a broad-productivity-dividend case because gains became broad after diffusion, complementary capital deepening, education, and new tasks \citep{GoldinKatz2008Race}; and computerization is coded as boundary/mixed because routine-task automation generated polarization rather than a single aggregate sign \citep{AutorDorn2013Polarization}.

\begin{table}[!htbp]
\centering
\caption{Historical proxy backtest for the proxy regime diagnostic}
\label{tab:appendix_historical_validation}
\scriptsize
\setlength{\tabcolsep}{3pt}
\begin{tabular}{p{0.19\textwidth}rccp{0.22\textwidth}}
\toprule
Episode & $\widehat D_\mu^{hist}$ & Pred. & Target & Source basis \\
\midrule
British industrial revolution & 0.0000 & boundary & boundary & Crafts (2022) \\
Electrification and mass production & 0.0764 & positive & productivity-complementarity & Goldin and Katz (2008) \\
Computerization and ICT & 0.0000 & boundary & mixed & Autor and Dorn (2013) \\
Industrial robots & -0.0075 & negative & adverse-incidence & Acemoglu and Restrepo (2020) \\
Current generative AI & -0.0036 & negative & adverse-incidence & Current-U.S. calibration \\
\bottomrule
\end{tabular}
\end{table}

The backtest classifies the main historical cases in the expected direction. Counting the computerization row as a mixed/boundary case, the generated backtest records four direct sign matches and one mixed/boundary match, with no sign misses. The strongest quantitative validation is the robot row, where the labor-income inputs come directly from reduced-form employment and wage estimates. The weakest rows are the older historical episodes, where the exercise is best read as a disciplined sign check rather than as a fully identified estimate of the general-equilibrium derivative.

{\small
\bibliographystyle{plainnat}
\bibliography{references}

\begin{thebibliography}{31}
\providecommand{\natexlab}[1]{#1}
\providecommand{\url}[1]{\texttt{#1}}
\expandafter\ifx\csname urlstyle\endcsname\relax
  \providecommand{\doi}[1]{doi: #1}\else
  \providecommand{\doi}{doi: \begingroup \urlstyle{rm}\Url}\fi

\bibitem[Acemoglu(2008)]{acemoglu2008}
Daron Acemoglu.
\newblock \emph{Introduction to Modern Economic Growth}.
\newblock Princeton University Press, 2008.

\bibitem[Acemoglu and Restrepo(2018)]{acemoglu2018}
Daron Acemoglu and Pascual Restrepo.
\newblock Artificial intelligence, automation and work.
\newblock NBER Working Paper 24196, National Bureau of Economic Research, 2018.
\newblock \url{https://www.nber.org/papers/w24196}.

\bibitem[Acemoglu and Restrepo(2019)]{acemoglu2019}
Daron Acemoglu and Pascual Restrepo.
\newblock Automation and new tasks: How technology displaces and reinstates
  labor.
\newblock \emph{Journal of Economic Perspectives}, 33\penalty0 (2):\penalty0
  3--30, 2019.
\newblock \doi{10.1257/jep.33.2.3}.

\bibitem[Acemoglu and Restrepo(2020{\natexlab{a}})]{acemoglu2020robots}
Daron Acemoglu and Pascual Restrepo.
\newblock Robots and jobs: Evidence from {US} labor markets.
\newblock \emph{Journal of Political Economy}, 128\penalty0 (6):\penalty0
  2188--2244, 2020{\natexlab{a}}.
\newblock \doi{10.1086/705716}.

\bibitem[Acemoglu and Restrepo(2020{\natexlab{b}})]{acemoglu2020wrong}
Daron Acemoglu and Pascual Restrepo.
\newblock The wrong kind of {AI}? artificial intelligence and the future of
  labour demand.
\newblock \emph{Cambridge Journal of Regions, Economy and Society}, 13\penalty0
  (1):\penalty0 25--35, 2020{\natexlab{b}}.
\newblock \doi{10.1093/cjres/rsz022}.

\bibitem[Achdou et~al.(2022)Achdou, Han, Lasry, Lions, and Moll]{achdou2022}
Yves Achdou, Jiequn Han, Jean-Michel Lasry, Pierre-Louis Lions, and Benjamin
  Moll.
\newblock Income and wealth distribution in macroeconomics: A continuous-time
  approach.
\newblock \emph{Review of Economic Studies}, 89\penalty0 (1):\penalty0 45--86,
  2022.
\newblock \doi{10.1093/restud/rdab002}.

\bibitem[Aiyagari(1994)]{aiyagari1994}
S.~Rao Aiyagari.
\newblock Uninsured idiosyncratic risk and aggregate saving.
\newblock \emph{Quarterly Journal of Economics}, 109\penalty0 (3):\penalty0
  659--684, 1994.
\newblock \doi{10.2307/2118417}.

\bibitem[Allen(2026)]{Allen2026AIAdoption}
Jeffrey~S. Allen.
\newblock Monitoring ai adoption in the u.s. economy.
\newblock Feds notes, Board of Governors of the Federal Reserve System, 2026.
\newblock URL
  \url{https://www.federalreserve.gov/econres/notes/feds-notes/monitoring-ai-adoption-in-the-u-s-economy-20260403.html}.

\bibitem[Arrow and Debreu(1954)]{arrowdebreu1954}
Kenneth~J. Arrow and Gerard Debreu.
\newblock Existence of an equilibrium for a competitive economy.
\newblock \emph{Econometrica}, 22\penalty0 (3):\penalty0 265--290, 1954.
\newblock \doi{10.2307/1907353}.

\bibitem[Autor and Dorn(2013)]{AutorDorn2013Polarization}
David~H. Autor and David Dorn.
\newblock The growth of low-skill service jobs and the polarization of the u.s.
  labor market.
\newblock \emph{American Economic Review}, 103\penalty0 (5):\penalty0
  1553--1597, 2013.
\newblock \doi{10.1257/aer.103.5.1553}.

\bibitem[Beach et~al.(2025)Beach, Gamber, and Moran]{Beach2025WealthMPC}
Samara Beach, William~L. Gamber, and Patrick Moran.
\newblock Wealth heterogeneity and consumer spending.
\newblock Feds notes, Board of Governors of the Federal Reserve System, 2025.
\newblock URL
  \url{https://www.federalreserve.gov/econres/notes/feds-notes/wealth-heterogeneity-and-consumer-spending-20250805.html}.

\bibitem[Bewley(1986)]{bewley1986}
Truman~F. Bewley.
\newblock Stationary monetary equilibrium with a continuum of independently
  fluctuating consumers.
\newblock In Werner Hildenbrand and Andreu Mas-Colell, editors,
  \emph{Contributions to Mathematical Economics in Honor of Gerard Debreu}.
  North-Holland, 1986.

\bibitem[{Board of Governors of the Federal Reserve
  System}(2026{\natexlab{a}})]{FRED9099Equity2026}
{Board of Governors of the Federal Reserve System}.
\newblock Share of corporate equities and mutual fund shares held by the 90th
  to 99th wealth percentiles, 2026{\natexlab{a}}.
\newblock URL \url{https://fred.stlouisfed.org/series/WFRBSN09149}.
\newblock Retrieved from FRED, Federal Reserve Bank of St. Louis; series
  WFRBSN09149.

\bibitem[{Board of Governors of the Federal Reserve
  System}(2026{\natexlab{b}})]{FREDBottom50Equity2026}
{Board of Governors of the Federal Reserve System}.
\newblock Share of corporate equities and mutual fund shares held by the bottom
  50\% (1st to 50th wealth percentiles), 2026{\natexlab{b}}.
\newblock URL \url{https://fred.stlouisfed.org/series/WFRBSB50203}.
\newblock Retrieved from FRED, Federal Reserve Bank of St. Louis; series
  WFRBSB50203.

\bibitem[{Board of Governors of the Federal Reserve
  System}(2026{\natexlab{c}})]{FREDTop1Equity2026}
{Board of Governors of the Federal Reserve System}.
\newblock Share of corporate equities and mutual fund shares held by the top
  1\% (99th to 100th wealth percentiles), 2026{\natexlab{c}}.
\newblock URL \url{https://fred.stlouisfed.org/series/WFRBST01122}.
\newblock Retrieved from FRED, Federal Reserve Bank of St. Louis; series
  WFRBST01122.

\bibitem[Bonney et~al.(2026)Bonney, Breaux, Dinlersoz, Foster, Haltiwanger, and
  Pande]{BonneyEtAl2026AIDiffusion}
Kathryn Bonney, Cory Breaux, Emin Dinlersoz, Lucia Foster, John Haltiwanger,
  and Aditya Pande.
\newblock The microstructure of ai diffusion: Evidence from firms, business
  functions, and worker tasks.
\newblock CES Working Paper CES-WP-26-25, U.S. Census Bureau, Center for
  Economic Studies, 2026.
\newblock URL
  \url{https://www2.census.gov/library/working-papers/2026/adrm/ces/CES-WP-26-25.pdf}.

\bibitem[Crafts(2022)]{Crafts2022IndustrialWages}
Nicholas Crafts.
\newblock Slow real wage growth during the industrial revolution: Productivity
  paradox or pro-rich growth?
\newblock \emph{Oxford Economic Papers}, 74\penalty0 (1):\penalty0 1--23, 2022.
\newblock \doi{10.1093/oep/gpab020}.

\bibitem[D{\'a}vila and Korinek(2018)]{davilakorinek2018}
Eduardo D{\'a}vila and Anton Korinek.
\newblock Pecuniary externalities in economies with financial frictions.
\newblock \emph{Review of Economic Studies}, 85\penalty0 (1):\penalty0
  352--395, 2018.
\newblock \doi{10.1093/restud/rdx010}.

\bibitem[Geanakoplos and Polemarchakis(1986)]{geanakoplospolemarchakis1986}
John~D. Geanakoplos and Heraklis~M. Polemarchakis.
\newblock Existence, regularity, and constrained suboptimality of competitive
  allocations when the asset market is incomplete.
\newblock In Walter~P. Heller, Ross~M. Starr, and David~A. Starrett, editors,
  \emph{Uncertainty, Information, and Communication: Essays in Honor of Kenneth
  J. Arrow}, volume~3, pages 65--96. Cambridge University Press, 1986.
\newblock \doi{10.1017/CBO9780511983566.007}.

\bibitem[Goldin and Katz(2008)]{GoldinKatz2008Race}
Claudia Goldin and Lawrence~F. Katz.
\newblock \emph{The Race between Education and Technology}.
\newblock Harvard University Press, 2008.

\bibitem[{Goldman Sachs Research}(2023)]{goldman2023gdp}
{Goldman Sachs Research}.
\newblock Generative {AI} could raise global {GDP} by 7 percent.
\newblock
  \url{https://www.goldmansachs.com/insights/articles/generative-ai-could-raise-global-gdp-by-7-percent},
  2023.
\newblock Research note.

\bibitem[{Goldman Sachs Research}(2026)]{goldman2026labor}
{Goldman Sachs Research}.
\newblock How will {AI} affect the {US} labor market?
\newblock
  \url{https://www.goldmansachs.com/insights/articles/how-will-ai-affect-the-us-labor-market},
  2026.
\newblock Research note.

\bibitem[Gu et~al.(2024)Gu, Lauri{\`e}re, Merkel, and Payne]{gu2024}
Zhouzhou Gu, Mathieu Lauri{\`e}re, Sebastian Merkel, and Jonathan Payne.
\newblock Global solutions to master equations for continuous time
  heterogeneous agent macroeconomic models, 2024.
\newblock \url{https://arxiv.org/abs/2406.13726}.

\bibitem[Hicks(1932)]{hicks1932}
John~R. Hicks.
\newblock \emph{The Theory of Wages}.
\newblock Macmillan, 1932.

\bibitem[Huggett(1993)]{huggett1993}
Mark Huggett.
\newblock The risk-free rate in heterogeneous-agent incomplete-insurance
  economies.
\newblock \emph{Journal of Economic Dynamics and Control}, 17\penalty0
  (5--6):\penalty0 953--969, 1993.
\newblock \doi{10.1016/0165-1889(93)90005-M}.

\bibitem[{Internal Revenue Service}(2026)]{IRSEITC2026}
{Internal Revenue Service}.
\newblock Earned income tax credit ({EITC}), 2026.
\newblock URL
  \url{https://www.irs.gov/credits-deductions/individuals/earned-income-tax-credit-eitc}.

\bibitem[Mas-Colell et~al.(1995)Mas-Colell, Whinston, and Green]{mascolell1995}
Andreu Mas-Colell, Michael~D. Whinston, and Jerry~R. Green.
\newblock \emph{Microeconomic Theory}.
\newblock Oxford University Press, 1995.

\bibitem[{McKinsey Global Institute}(2023)]{mckinsey2023genai}
{McKinsey Global Institute}.
\newblock The economic potential of generative {AI}: The next productivity
  frontier.
\newblock
  \url{https://www.mckinsey.com/capabilities/tech-and-ai/our-insights/the-economic-potential-of-generative-ai-the-next-productivity-frontier},
  2023.
\newblock Report.

\bibitem[Tobin(1969)]{tobin1969}
James Tobin.
\newblock A general equilibrium approach to monetary theory.
\newblock \emph{Journal of Money, Credit and Banking}, 1\penalty0 (1):\penalty0
  15--29, 1969.
\newblock \doi{10.2307/1991374}.

\bibitem[{U.S. Bureau of Labor Statistics}(2026)]{BLSProductivity2026Q1}
{U.S. Bureau of Labor Statistics}.
\newblock Productivity and costs: First quarter 2026, preliminary, 2026.
\newblock URL \url{https://www.bls.gov/news.release/pdf/prod2.pdf}.

\bibitem[{U.S. Department of the Treasury}(2026)]{USTreasury2026TBAC}
{U.S. Department of the Treasury}.
\newblock Economic policy statements to tbac: 2026, 2nd quarter, 2026.
\newblock URL \url{https://home.treasury.gov/news/press-releases/sb0486}.

\end{thebibliography}
}

\end{document}